\documentclass[a4paper,english]{lipics-v2018}
\usepackage{xspace}
\usepackage{graphicx, wrapfig}
\usepackage{tabularx}
\usepackage{booktabs}
\usepackage{comment}
\usepackage{scalerel}
\usepackage[usenames,dvipsnames]{xcolor}
\usepackage[nocompress]{cite}

\graphicspath{{figures/}}

\usepackage{microtype}

\newtheorem{observation}[theorem]{\textbf{Observation}}

\newcommand{\thmheadfont}{\textcolor{darkgray}{$\blacktriangleright$}\nobreakspace\sffamily\bfseries}
\newenvironment{repeatenv}[2]%
  {\smallskip\noindent {\thmheadfont #1~\ref{#2}.}\ \slshape}
  {\normalfont}

\newenvironment{repeattheorem}    [1]{\begin{repeatenv}{Theorem}{#1}}    {\end{repeatenv}}

\newcommand{\proofcmd}{\noindent{\color{darkgray}\sffamily\bfseries Proof. }}

\newcommand{\mkmcal}[1]{\ensuremath{\mathcal{#1}}\xspace}

\newcommand{\A}{\mkmcal{A}}
\newcommand{\RR}{\mkmcal{R}}

\newcommand{\VD}{\mkmcal{V}}
\newcommand{\F}{\mkmcal{F}}

\newcommand{\D}{\mkmcal{D}}

\renewcommand{\S}{\mkmcal{S}}
\renewcommand{\P}{\mkmcal{P}}

\newcommand{\mkmbb}[1]{\ensuremath{\mathbb{#1}}\xspace}

\DeclareMathOperator{\polylog}{polylog}

\newcommand{\R}{\mkmbb{R}}

\newcommand{\eps}{\varepsilon}

\newcommand{\etal}{et al.\xspace}

\newcommand{\br}{b^*}

\newcommand{\geod}{\Pi\xspace}
\newcommand{\geodlen}{\pi\xspace}

\newcommand{\acall}[2]{\textsc{#1}(#2)\xspace}

\title{Improved Dynamic Geodesic Nearest Neighbor Searching in a Simple Polygon}

\author{Pankaj K. Agarwal}{Department of Computer Science, Duke University,
  \\{Durham, NC 27708, USA}}{pankaj@cs.duke.edu}{}{P.A. was supported by NSF under
grants CCF-15-13816, CCF-15-46392, and IIS-14-08846, by ARO under grant
W911NF-15-1-0408, and by grant 2012/229 from the U.S.--Israel Binational Science Foundation}
\author{Lars Arge}{MADALGO, Aarhus University,\\ {Aarhus,
    Denmark}}{large@cs.au.dk}{}{L.A. was supported by the Danish National Research Foundation under grant nr.~DNRF84}
\author{Frank Staals}{Dept. of Information and Computing Sciences, Utrecht
  University, \\{Utrecht, The Netherlands}}{f.staals@uu.nl}{}{F.S. was
  supported by the Netherlands Organisation for Scientific Research (NWO) under
  project no. 612.001.651.}

\authorrunning{P. K. Agarwal, L. Arge, F. Staals} 

\Copyright{Pankaj K. Agarwal, Lars Arge, Frank Staals}

\subjclass{Computational Geometry}


\category{}

\EventEditors{Bettina Speckmann and Csaba D. T{\'o}th}
\EventNoEds{2}
\EventLongTitle{34th International Symposium on Computational Geometry (SoCG 2018)}
\EventShortTitle{SoCG 2018}
\EventAcronym{SoCG}
\EventYear{2018}
\EventDate{June 11--14, 2018}
\EventLocation{Budapest, Hungary}
\EventLogo{socg-logo} 
\SeriesVolume{99}
\ArticleNo{4}

\nolinenumbers 
\hideLIPIcs  

\begin{document}

\maketitle

\begin{abstract}
  We present an efficient dynamic data structure that supports geodesic nearest
  neighbor queries for a set $S$ of point sites in a static simple polygon
  $P$. Our data structure allows us to insert a new site in $S$, delete a site
  from $S$, and ask for the site in $S$ closest to an arbitrary query point
  $q \in P$. All distances are measured using the geodesic distance, that is,
  the length of the shortest path that is completely contained in $P$.  Our
  data structure achieves polylogarithmic update and query times, and uses
  $O(n\log^3n\log m + m)$ space, where $n$ is the number of sites in $S$ and
  $m$ is the number of vertices in $P$. The crucial ingredient in our data
  structure is an implicit representation of a vertical shallow cutting of the
  geodesic distance functions. We show that such an implicit representation
  exists, and that we can compute it efficiently.
\end{abstract}

\section{Introduction}
\label{sec:Introduction}

Nearest neighbor searching is a classic problem in computational geometry in
which we are given a set of point \emph{sites} $S$, and we wish to preprocess
these points such that for a query point $q$, we can efficiently find the site
$s \in S$ closest to $q$. We consider the case where $S$ is a \emph{dynamic}
set of points inside a simple polygon $P$. That is, we may insert a new site
into $S$ or delete an existing one. We measure the distance between two points
$p$ and $q$ by the length of the \emph{geodesic} $\geod(p,q)$, that is, the
shortest path connecting $p$ and $q$ that is completely contained in $P$. We
refer to this distance as the \emph{geodesic distance} $\geodlen(p,q)$.

\subparagraph{Related work.} It is well known that if we have only a fixed set
$S$ of $n$ sites, we can answer nearest neighbor queries efficiently by
computing the Voronoi diagram of $S$ and preprocessing it for planar point
location. This requires $O(n\log n)$ preprocessing time, the resulting data
structure uses linear space, and we can answer queries in $O(\log n)$ time. Voronoi
diagrams have also been studied in case the set of sites is restricted to lie
in a simple polygon $P$, and we measure the distance between two points $p$ and
$q$ by their geodesic distance
$\geodlen(p,q)$~\cite{aronov1989geodesic,papadopoulou1998geodesic,hershberger1999sssp,oh_ahn2017voronoi}. The
approach of Hershberger and Suri~\cite{hershberger1999sssp} computes the
geodesic Voronoi diagram in $O((m+n)\log(m+n))$ time, where $m$ is the total
number of vertices in the polygon $P$, and is applicable even if $P$ has
holes. Very recently, Oh and Ahn~\cite{oh_ahn2017voronoi} presented an
$O(m+n\log n\log^2 m)$ time algorithm. When $n\leq m/\log^3 m$ this improves
the previous results. All these approaches allow for $O(\log(n+m))$ time
nearest neighbor queries. However, they are efficient only when the set of
sites $S$ is fixed, as inserting or deleting even a single site may cause a
linear number of changes in the Voronoi diagram.

To support nearest neighbor queries, it is, however, not necessary to
explicitly maintain the (geodesic) Voronoi diagram. Bentley and
Saxe~\cite{bentley1980decomposable} show that nearest neighbor searching is a
\emph{decomposable search problem}. That is, we can find the answer to a query
by splitting $S$ into groups, computing the solution for each group
individually, and taking the solution that is best over all groups. This
observation has been used in several other approaches for nearest neighbor
searching with the Euclidean
distance~\cite{agarwal1995dynamic,dobkin1991maintenance,chan2010dynamic_ch}. However,
even with this observation, it is hard to get both polylogarithmic update and
query time. Chan~\cite{chan2010dynamic_ch} was the first to achieve this. His
data structure can answer Euclidean nearest neighbor queries in $O(\log^2 n)$
time, and supports insertions and deletions in $O(\log^3 n)$ and $O(\log^6 n)$
amortized time, respectively. Recently, Kaplan \etal~\cite{kaplan2017dynamic}
extended the result of Chan to more general, constant complexity, distance
functions.





Unfortunately, the above results do not directly lead to an efficient solution
to our problem. The function describing the geodesic distance may have
complexity $\Theta(m)$, and thus the results of Kaplan
\etal~\cite{kaplan2017dynamic} do not apply. Moreover, even directly combining
the decomposable search problem approach with the static geodesic Voronoi
diagrams described above does not lead to an efficient solution, since every
update incurs an $\Omega(m)$ cost corresponding to the complexity of the
polygon. Only the very recent algorithm of Oh and Ahn~\cite{oh_ahn2017voronoi}
can be made amendable to such an approach. This results in an $O(n+m)$ size
data structure with $O(\sqrt{n}(\log n +\log m))$ query time and
$O(\sqrt{n}\log n\log^2 m)$ updates. Independently from Oh and Ahn, we
developed a different data structure yielding similar
results~\cite{arge2017arxiv}. The core idea in both data structures is to
represent the Voronoi diagram implicitly. Moreover, both approaches use similar
primitives. In this manuscript, we build on the ideas from our earlier work,
and significantly extend them to achieve polylogarithmic update and query
times.



\enlargethispage{\baselineskip}
\subparagraph{Our results.} We develop a fully dynamic data structure to
support nearest neighbor queries for a set of sites $S$ inside a (static)
simple polygon $P$. Our data structure allows us to locate the site in $S$
closest to a query point $q \in P$, to insert a new site $s$ into $S$, and to
delete a site from $S$. Our data structure supports queries in
$O(\log^2n\log^2 m)$ time, insertions in $O(\log^5n\log m + \log^4n\log^3 m)$
amortized expected time, and deletions in $O(\log^7n\log m + \log^6 n\log^3 m)$
amortized expected time. The space usage is $O(n\log^3 n\log m + m)$.

Furthermore, we show that using a subset of the tools and techniques that we
develop, we can build an improved data structure for when there are no
deletions. In this insertion-only setting, queries take worst-case
$O(\log^2 n\log^2 m)$ time, and insertions take amortized $O(\log n\log^3 m)$
time. We can also achieve these running times in case
there are both insertions and deletions, but the order of these operations is
known in advance. The space usage of this version is $O(n\log n\log m + m)$.

\section{An overview of the approach}
\label{sec:overview}

As in previous work on geodesic Voronoi diagrams~\cite{aronov1989geodesic,
  papadopoulou1998geodesic}, we assume that $P$ and $S$ are in general
position. That is, (i) no two sites $s$ and $t$ in $S$ (ever) have the same
geodesic distance to a vertex of $P$, and (ii) no three points (either sites or
vertices) are colinear. Note that (i) implies that no bisector $b_{st}$ between
sites $s$ and $t$ contains a vertex of $P$.

Throughout the paper we will assume that the input polygon $P$ has been
preprocessed for two-point shortest path queries using the data structure by
Guibas and Hershberger~\cite{guibas1989query} (see also the follow up note of
Hershberger~\cite{hershberger_new_1991}). This takes $O(m)$ time and allows us
to compute the geodesic distance $\geodlen(p,q)$ between any pair of query
points $p,q \in P$ in $O(\log m)$ time.

\subparagraph{Dynamic Euclidean nearest neighbor searching.} We briefly review
the data structures for dynamic Euclidean nearest neighbor searching of
Chan~\cite{chan2010dynamic_ch} and Kaplan \etal~\cite{kaplan2017dynamic}, and
the concepts they use, as we will build on these results.

Let $F$ be a set of bivariate functions, and let $\A(F)$ denote the arrangement
of the (graphs of the) functions in $\R^3$. In the remainder of the paper we
will no longer distinguish between a function and its graph. A point
$q \in \R^3$ has \emph{level} $k$ if the number of functions in $F$ that pass
strictly below $q$ is $k$. The \emph{at most $k$-level} $L_{\leq k}(F)$ is the
set of all points in $\R^3$ for which the level is at most $k$, and the
$k$-level $L_k(F)$ is the boundary of that region.

Consider a collection $X$ of points in $\R^3$ (e.g.~a line segment), and let
$F_X$ denote the set of functions from $F$ intersecting $X$. We refer to $F_X$
as the \emph{conflict list} of $X$. Furthermore, let $\underline{X}$ denote the
vertical (downward) projection of $X$ onto the $x,y$-plane.

A \emph{pseudo-prism} is a constant complexity region in $\R^3$ that is bounded
from above by a function, unbounded from below, and whose sides are surfaces
vertical with respect to the $z$-direction. A \emph{a $k$-shallow
  $(1/r)$-cutting} $\Lambda_{k,r}(F)$ is a collection of such pseudo-prisms
with pairwise disjoint interiors whose union covers $L_{\leq k}(F)$ and for
which each pseudo-prism intersects at most $n/r$ functions in
$F$~\cite{m-rph-92}. Hence, for each region (pseudo-prism)
$\nabla \in \Lambda_{k,r}(F)$, the conflict list $F_\nabla$ contains $n/r$
functions. The number of regions in $\Lambda_{k,r}(F)$ is the \emph{size} of
the cutting.  Matou{\v{s}}ek~\cite{m-rph-92} originally defined a shallow
cutting in terms of simplicies, however using pseudo-prisms is more convenient
in our setting. In this case the parameter $r$ cannot become arbitrarily large,
however we are mostly interested in $k$-shallow $O(k/n)$-cuttings. In the
remainder of the paper we simply refer to such a cutting $\Lambda_k(F)$ as
\emph{a $k$-shallow cutting}. Observe that each pseudo-prism in $\Lambda_k(F)$
is intersected by $O(k)$ functions.



\enlargethispage{\baselineskip}
Let $f_s(x)$ be the distance from $x$ to $s$. The data structures of Kaplan
\etal~\cite{kaplan2017dynamic} and Chan~\cite{chan2010dynamic_ch} actually
maintain the \emph{lower envelope} $L_0(F)$ of the set of distance functions
$F = \{ f_s \mid s \in S\}$. To find the site $s$ closest to a query point $q$
they can simply query the data structure to find the function $f_s$ that
realizes $L_0(F)$ at $q$. The critical ingredient in both data structures, as
well as our own data structure, is an efficient algorithm to construct a
shallow cutting of the functions in $F$.
Chan and Tsakalidis~\cite{chan_tsakalidis2015shallow} show that if $F$ is a set
of linear functions (i.e. planes in $\R^3$), we can compute a $k$-shallow
cutting $\Lambda_k(F)$ of size $O(n/k)$ in $O(n\log n)$ time. Kaplan
\etal~\cite{kaplan2017dynamic} show how to compute a $k$-shallow cutting of
size $O(n\log^2 n/k)$, in time $O(n\polylog n)$ for a certain class of constant
complexity algebraic functions $F$. Since the geodesic distance function
$f_s(x) = \geodlen(s,x)$ of a single site $s$ may already have complexity
$\Theta(m)$, any $k$-shallow cutting of such functions may have size
$\Omega(m)$. To circumvent this issue, we allow the regions (pseudo-prisms) to
have non-constant complexity. We do this in such a way that we can compactly
represent each pseudo-prism, while still retaining some nice properties such as
efficiently testing if a point lies inside it. Hence, in the remainder of the
paper we will drop the requirement that the regions in a cutting need to have
constant complexity.

\subparagraph{The general approach.} The general idea in our approach is to
recursively partition the polygon into two roughly equal size sub-polygons
$P_\ell$ and $P_r$ that are separated by a
diagonal. For the sites $S_\ell$ in the ``left''
subpolygon $P_\ell$, we then consider their geodesic distance functions
$F_\ell=\{f_s \mid s \in S_\ell\}$ restricted to the ``right'' subpolygon
$P_r$, that is, $f_s(x)=\geodlen(s,x)$, for $x \in P_r$. The crucial part, and our
main contribution, is that for these functions $F_\ell$ we can represent a
vertical shallow cutting implicitly. See Fig.~\ref{fig:global} for a schematic
illustration. More specifically, in
$O((n/k)\log^3 n(\log n + \log^2 m) + n\log^2 m + n\log^3 n\log m)$ expected
time, we can build a representation of the shallow cutting of size
$O((n/k)\log^2 n)$. We can then use this algorithm for building implicitly
represented shallow cuttings in the data structure of
Chan~\cite{chan2010dynamic_ch} and Kaplan~\etal~\cite{kaplan2017dynamic}. That
is, we build and maintain the lower envelope $L_0(F_\ell)$.  Symmetrically,
for the sites in $P_r$, we maintain the lower envelope $L_0(F_r)$ that their
distance functions $F_r$ induce in $P_\ell$. When we get a query point
$q \in P_r$, we use $L_0(F_\ell)$ to find the site in $P_\ell$ closest to $q$
in $O(\log^2 n\log m)$ time. To find the site in $P_r$ closest to $q$, we
recursively query in sub-polygon $P_r$. In total we query in $O(\log m)$
levels, leading to an $O(\log^2n\log^2 m)$ query time. When we add or remove a
site $s$ we, insert or remove its distance function in $O(\log m)$ lower
envelope data structures (one at every level). Every insertion takes
$O(\log^5 n + \log^4n\log^2 m)$ amortized expected time, and every deletion
takes $O(\log^7 n+\log^6\log^2 m)$ amortized expected time. Since every site is
stored $O(\log m)$ times, this leads to the following main result.

\begin{theorem}
  \label{thm:fully_dynamic_ds_polylog}
  Let $P$ be a simple polygon $P$ with $m$ vertices. There is a fully dynamic data
  structure of size $O(n\log^3 n\log m + m)$ that maintains a set of
  $n$ point sites in $P$ and allows for geodesic nearest neighbor queries in
  worst case $O(\log^2 n\log^2 m)$ time. Inserting a site takes
  $O(\log^5 n\log m + \log^4 n\log^3 m)$ amortized expected time, and deleting
  a site takes $O(\log^7n\log m + \log^6 n\log^3 m)$ amortized expected time.
\end{theorem}

\begin{figure}[tb]
  \centering
  \includegraphics[page=2]{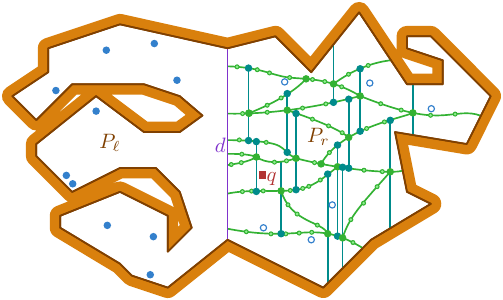}
  \caption{A schematic drawing of the downward projection of an implicit
    $k$-shallow cutting $\underline{\Lambda_k(F_\ell)}$ in $P_r$. The faces are
    pseudo-trapezoids. Neighboring pseudo-trapezoids share a vertical segment,
    or part of a bisector. We store only the degree one and three vertices
    (fat), and their topology.}
  \label{fig:global}
\end{figure}

The main complexity is in developing our implicit representation of the
$k$-shallow cutting, and the algorithm to construct such a cutting. Once we
have this algorithm we can directly plug it in into the data structure of
Chan~\cite{chan2010dynamic_ch} and Kaplan~\etal~\cite{kaplan2017dynamic}. Our global
strategy is similar to that of Kaplan~\etal~\cite{kaplan2017dynamic}: we
compute an approximate $k$-level of $\A(F)$ --in our case an implicit
representation of this approximate $k$-level-- and then argue that, under
certain conditions, this approximate $k$-level is actually a $k$-shallow
cutting $\Lambda_k(F)$ of $\A(F)$. Our approximate $k$-level will be a
$t$-level, for some appropriate $t$, on a random sample of the functions in
$F$. So, that leaves us two problems: \textit{i}) computing an implicit
representation of a $t$-level, and \textit{ii}) computing the conflict lists
for all pseudo-prism in our cutting $\Lambda_k(F)$.

For problem \textit{i}), representing the $t$-level implicitly, we use the
connection between the $t$-level and the $t^\mathrm{th}$-order Voronoi
diagram. In Section~\ref{sec:implicit_representations} we describe a small,
implicit representation of the $t^\mathrm{th}$-order Voronoi diagram that still
allows us to answer point location queries efficiently. Initially, we use the
recent algorithm of Oh and Ahn~\cite{oh_ahn2017voronoi} to construct this
representation. In Section~\ref{sec:computing_implict_representations} we then
design an improved algorithm for our particular use case.

For problem \textit{ii}), computing the conflict lists, we will use a data
structure developed by Chan~\cite{chan2000rangerep} together with the implicit
Voronoi diagrams that we developed in
Section~\ref{sec:computing_implict_representations}. We describe these results
in more detail in Section~\ref{sec:implicit_cutting}. In
Section~\ref{sec:combining}, we then show in detail how we can combine all the
different parts into a fully dynamic data structure for nearest neighbor
queries.  Finally, we describe our simpler data structures for insertion-only
and offline-updates in Section~\ref{sec:simplified}.

\section{Implicit representations}
\label{sec:implicit_representations}

Let $F=\{f_s \mid s \in S\}$ denote the set of geodesic distance functions
inside the entire polygon $P$. Our implicit representation of a $k$-shallow
cutting $\Lambda_k(F)$ is based on an implicit representation of the $k$-level
in $\A(F)$. To this end, we first study higher order geodesic Voronoi diagrams.

\subparagraph{Higher order Voronoi diagrams.} Consider a set of $n$ sites $S$,
a domain $\D$, and a subset $H \subseteq S$ of size $k$. The
\emph{$k^\mathrm{th}$-order Voronoi region} $V_k(H,S)$ is the region in \D in
which the points are closer to (a site in) $H$, with respect to some distance
metric, than to any other subset $H' \subseteq S$ of size $k$. The
\emph{$k^\mathrm{th}$-order Voronoi diagram} $\VD_k(S)$ is the partition of \D
into such maximal regions $V_k(H,S)$ over all subsets $H \subseteq S$ of size
$k$~\cite{lee1982kthorder_vd}. Liu \etal~\cite{liu2013geodesic_higher_order_vd}
study the \emph{geodesic} $k^\mathrm{th}$-order Voronoi diagram. They show that
$\VD_k(S)$ has complexity $O(k(n-k)+km)$. In particular, it consists of
$O(k(n-k))$ degree one and degree three vertices, and $O(km)$ degree two
vertices (and by our general position assumption, there are no vertices of
degree more than three).

Consider a Voronoi region $V_k(H,S)$. Let $e_1,..,e_\ell$, be the edges
bounding $\partial V_k(H,S)$, let $H_j$ be the set of sites defining the other
Voronoi region incident to $e_j$, and observe that $H_j \setminus H$ contains a
single site $q_j$~\cite{lee1982kthorder_vd}. Let
$Q = \{ q_j \mid j \in [1,\ell]\}$ be the set of sites \emph{neighboring}
$V_k(H,S)$. Observe that these results imply that two adjacent regions
$V_k(H,S)$ and $V_k(H_j,S)$ in $\VD_k(S)$ are separated by a part of a bisector
$b_{st}$, where $s = H\setminus H_j$ and $t=q_j$.

By combining the above two observations we can represent $\VD_k(S)$
implicitly. That is, we store only the locations of these degree one and degree three
vertices, the adjacency relations between the regions, and the pair of labels
$(s,t)$ corresponding to each pair $(R_s,R_t)$ of neighboring regions. See also
the recent result of Oh and Ahn~\cite{oh_ahn2017voronoi}. It follows that the
size of this representation is linear in the number of degree one and degree
three vertices of $\VD_k(S)$. We refer to this as the \emph{topological
  complexity} of $\VD_k(S)$.

\subparagraph{Representing the $k$-level.} Consider the partition of $P$ into
maximally connected regions in which all points in a region have the same
$k^\mathrm{th}$ nearest site in $S$. Observe that this partition corresponds to
the downward projection $\underline{L_k(F)}$ of the $k$-level $L_k(F)$. As we
argue next, this partition is closely related to the $k^\mathrm{th}$-order
Voronoi diagram defined above.

Lee observes that there is a relation between the $i^\mathrm{th}$-order Voronoi
diagram and the $(i+1)^\mathrm{th}$-order Voronoi
diagram~\cite{lee1982kthorder_vd}. In particular, he shows that we can
partition each $i^\mathrm{th}$-order Voronoi region $V_i(H,S)$ into
$(i+1)^\mathrm{th}$ order Voronoi regions by intersecting $V_i(H,S)$ with the
(first order) Voronoi diagram of the set of sites neighboring $V_i(H,S)$. More
specifically:

\begin{observation}
  \label{obs:k_nearest_site}
  Let $V_i(H,S)$ be a geodesic $i^\mathrm{th}$-order Voronoi region, and let
  $Q$ be the sites neighboring $V_i(H,S)$. For any point $p \in V_i(H,S)$, the
  $(i+1)$-closest site from $p$ is the site $s \in Q$ for which
  $p \in \VD(s,Q)$.
\end{observation}



\begin{lemma}
  \label{lem:projected_k_level}
  The topological complexity of $\underline{L_k(F)}$  is $O(k(n-k))$.
\end{lemma}

\begin{proof}
  By Observation~\ref{obs:k_nearest_site}, we can obtain $\underline{L_k(F)}$
  from $\VD_{k-1}(S)$ by partitioning every region $V_{k-1}(H,S)$ using the
  Voronoi diagram of the sites $Q$ neighboring $V_{k-1}(H,S)$, and merging
  regions that have the same $k^\mathrm{th}$ nearest neighbor in the resulting
  subdivision. Thus, the vertices in $\underline{L_k(F)}$ appear either in
  $V_{k-1}(H,S)$ or in one of the newly created Voronoi diagrams.

  There are $O((k-1)(n-k+1))=O(k(n-k))$ degree one and degree three vertices
  that are also vertices in $V_{k-1}(H,S)$. Since the (first order) geodesic
  Voronoi diagram has $O(k)$ degree one and degree three vertices, a region
  $V_{k-1}(H,S)$ creates $|Q|$ new degree one and three vertices. Summing over
  all faces in $\VD_{k-1}(H,S)$ this then sums to $O(k(n-k))$.
\end{proof}

As with $\VD_k(S)$ we can represent $\underline{L_k(F)}$ implicitly by storing
only the locations of the degree one and three vertices and the topology of the
diagram. Note that if we can efficiently locate the region $R_s$ of
$\underline{L_k(F)}$ that contains a query point $q$, we can also easily
compute the $z$-coordinate of $L_k(F)$ at $q$, simply by computing the geodesic
distance $\geodlen(s,q)$. Since we preprocessed $P$ for two-point shortest path
queries this takes only $O(\log m)$ time (in addition to locating the region of
$\underline{L_k(F)}$ containing $q$).

\begin{wrapfigure}[11]{r}{0.4\textwidth}
  \centering
  \includegraphics{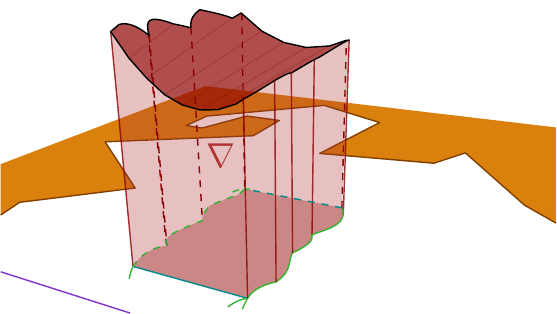}
  \caption{A pseudo prism $\nabla$ and its projection $\underline{\nabla}$
    (darker red) onto $P_r$.}
  \label{fig:geod_pseudo-prism}
\end{wrapfigure}

\subparagraph{Representing a vertical decomposition.} For every degree three
and degree one vertex in $\underline{L_k(F)}$ we now extend a vertical segment
up and down until it hits another edge of $\underline{L_k(F)}$. Observe that
the topological complexity of the resulting \emph{implicit vertical
  decomposition} $\underline{L_k^\nabla(F)}$ that we obtain still has
topological complexity $O(k(n-k))$. Furthermore, each region in
$\underline{L_k^\nabla(F)}$ is a \emph{pseudo-trapezoid} $\underline{\nabla}$
that is bounded on the left and right either by a vertical segment or a piece
of polygon boundary, and on the top and bottom by pieces of bisectors or a
piece of polygon boundary. We refer to the four degree one or degree three
vertices on the boundary of $\underline{\nabla}$ as the \emph{corners} of
$\underline{\nabla}$. See Fig.~\ref{fig:geod_pseudo-prism} for an
illustration. In the remainder of the paper, we will no longer distinguish
between $\underline{L_k^\nabla(F)}$ and its implicit representation. In
Section~\ref{sec:implicit_cutting}, we will use such an implicit vertical
decomposition to obtain an implicit representation of a shallow cutting.

\subparagraph{Computing implicit representations.} We can use the algorithm of
Oh and Ahn~\cite{oh_ahn2017voronoi} to compute the implicit representation of
$\VD_k(S)$ in $O(k^2n\log n\log^2 m)$ time. To compute (the representation of)
$\underline{L_k(F)}$ we first construct $\VD_{k-1}(S)$, and for each region
$V_{k-1}(H,S)$ in $\VD_{k-1}(S)$, we again use their algorithm to compute the
Voronoi diagrams $\VD(Q)$ of the set neighbors $Q$. We then clip these diagrams
to $V_{k-1}(H,S)$. This clipping can be done by a breadth first search in
$\VD(Q)$, starting with one of the vertices that is also in
$\partial V_{k-1}(H,S)$. This takes $O(|Q|\log |Q|\log^2 m)$
time in total. Summing over all faces in $\VD_{k-1}(S)$ gives us again a running time of
$O(k^2n\log n\log^2 m)$. Finally, to compute $\underline{L_k^\nabla(F)}$ we need to
insert two vertical extension segments at each vertex of $\underline{L_k(F)}$. We
can find the other endpoint of each extension segment using a point location
query. Thus, we obtain the following result.

\begin{lemma}
  \label{lem:implicit_representation}
  An implicit representation $\underline{L_k^\nabla(F)}$ of the $k$-level
  $L_k(F)$ that uses $O(k(n-k))$ space, can be computed in $O(k^2n\log n\log^2 m)$
  time. Using this representation, the pseudo-prism containing
  a query point (if it exists) can be determined in $O(\log n + \log m)$ time.
\end{lemma}

In Section~\ref{sec:computing_implict_representations} we will show that if the
sites defining the functions in $F$ lie in one half of the polygon and we
restrict the functions to the other half we can improve these
results. Moreover, we can then compute an implicit representation of a
$k$-shallow cutting of $F$.

\section{Approximating the $k$-level}
\label{sec:Approximating_the_k-Level}




An $xy$-monotone surface $\Gamma$ is an \emph{$\eps$-approximation} of $L_k(F)$
if and only if $\Gamma$ lies between $L_k(F)$ and $L_{(1+\eps)k}(F)$. Following
the same idea as in Kaplan \etal~\cite{kaplan2017dynamic} we construct an
$\eps$-approximation of $L_k(F)$ as follows. We choose a random sample $R$ of
$F$ of size $r=(cn/k\eps^2)\log n$ and consider the $t$-level of $\A(R)$ for
some randomly chosen level $t$ in the range $[(1+\eps/3)h,(1+\eps/2)h]$. Here
$c$ and $c'$ are some constants, $h = c'/\eps^2\log n$, and $\eps \in [0,1/2]$
is the desired approximation ratio. We now argue that $L_t(R)$ is an
$\eps$-approximation of $L_k(F)$.

Consider the range space $\S=(F,\RR)$, where each range in $\RR$ is the subset
of functions of $F$ intersected by a downward vertical ray in the
$(-z)$-direction. See Har-Peled~\cite{harPeled2011geometric} for details on
range spaces. An important concept is the \emph{VC-dimension} of \S, defined as
the size of the largest subset $F' \subseteq F$ for which the number of sets in
$\{ F' \cap F_\rho \mid F_\rho \in \RR\}$ is $2^{|F'|}$.

\begin{lemma}[Lemma 2.3.5 of Aronov~\etal\cite{aronov1993furthest}]
  \label{lem:intersect_once}
  Let $s$, $t$, and $u$ be three sites in $P$. Their bisectors $b_{st}$ and
  $b_{tu}$ intersect in at most a single point.
\end{lemma}

\begin{lemma}
  \label{lem:finite_vc_dimension}
  The VC-dimension of the range space \S is finite.
\end{lemma}

\begin{proof}
  The range space $(F,\RR)$ is equivalent to $(S_\ell, \D)$ where
  $\D \subseteq 2^{S_\ell}$ is the family of subsets of $S_\ell$ that lie
  within some geodesic distance of some point $p \in P$. That is
  $\D = \{ \{ s \in S_\ell \mid \geodlen(p,s) \leq z \} \mid p\in P, z \geq
  0\}$. We now observe that any three points $s$, $t$, and $u$, define at most
  one geodesic disk (Lemma~\ref{lem:intersect_once}). Namely, the disk centered
  at the intersection point $p=b_{st} \cap b_{tu}$ and radius
  $\geodlen(p,s)=\geodlen(p,t)=\geodlen(p,u)$. This means the same argument
  used by as Har-Peled~\cite[Lemma 5.15]{harPeled2011geometric} now gives us
  that the shattering dimension of \S (see~\cite{harPeled2011geometric}) is
  constant (three). It then follows that the VC-dimension of \S is finite.
\end{proof}

Since \S has finite VC-dimension (Lemma~\ref{lem:finite_vc_dimension}), and $R$
has size
$r=(cn/k\eps^2)\log n \geq
\left(\frac{c}{\eps^2p}\left(\log\frac{1}{p}+\log\frac{1}{q}\right)\right)$,
for $p=\frac{k}{2n}$ and $q=1/n^b$ for some sufficiently large constant $b$, it
follows that with high probability, $R$ is a \emph{relative
  $(p,\frac{\eps}{3})$-approximation} for
$\S=(F,\RR)$~\cite{harPeled2011relative}. So, for every range $H \in \RR$, we
have (whp.) that
\begin{equation}
  \left| \frac{|H|}{|F|}-\frac{|H \cap R|}{|H|}
  \right|
  \leq
  \begin{cases}
    \frac{\eps}{3}\frac{|H|}{|F|}, & \text{if } |H| \geq p|F|, \text{ and }\\
    \frac{\eps}{3}p                & \text{if } |H| <    p|F|.
  \end{cases}
\end{equation}
Using exactly the same argument as Kaplan \etal~\cite{kaplan2017dynamic} we
then obtain the following result.

\begin{lemma}
  \label{lem:sample_is_approx_level}
  The level $L_t(R)$ is an $\eps$-approximation of the $k$-level $L_k(F)$.
\end{lemma}

What remains is to show that the (expected) topological complexity of the
$t$-level $L_t(R)$ is small, that is, that the expected number of degree three
and degree one vertices is at most $O((n/k\eps^5)\log^2 n)$.

\begin{lemma}
  \label{lem:complexity_level}
  The expected topological complexity of level $L_t(R)$ is
  $O((n/k\eps^5)\log^2 n)$.
\end{lemma}

\begin{proof}
  The lower envelope $L_0(R)$ of $R$ has topological complexity $O(r)$, so by
  Clarkson and Shor~\cite{cs-arscg-89} the total topological complexity of all
  levels in the range $[(1+\eps/3)h,(1+\eps/2)h]$ is $O(rh^2)$. Hence, the
  expected topological complexity of a level $L_t(R)$, with $t$ randomly chosen
  from this range, is $O(rh^2/h\eps)=O(rh/\eps)$. Substituting
  $r=(cn/k\eps^2)\log n$ and $h=c'/\eps^2\log n$, we get $O(nk/\eps^5\log^2n)$
  as claimed.
\end{proof}

Again as in Kaplan \etal~\cite{kaplan2017dynamic}, if $k < (1/\eps^2)\log n$,
we can skip the sampling of the set $R$, and directly take a random level $t$
in $\A(F)$ in the range $[k,k(1+\eps)]$. This level has also an expected
topological complexity of at most $O((n/k\eps^5)\log^2 n)$. Using
Lemma~\ref{lem:implicit_representation} (and restarting the computation if the
size of the cutting exceeds $O((n/k\eps^5)\log^2 n)$)
we then get:

\begin{lemma}
  \label{lem:approx_k_level}
  An $\eps$-approximation of the $k$-level of $\A(F)$ that has topological
  complexity $O((n/k\eps^5)\log^2 n)$ can be computed in expected
  $O((n/k\eps^6)\log^3 n\log^2m)$ time.
\end{lemma}

The main difference between our approach and that of
Kaplan~\etal~\cite{kaplan2017dynamic} is the range space used. In our approach,
the ranges are defined by downward vertical rays, whereas in Kaplan~\etal the
ranges are defined by more general objects. For example, their range space
includes a range consisting of the functions intersected by some other constant
complexity algebraic function. This allows them to directly turn their
approximate level into a shallow cutting. Unfortunately, this idea does not
directly extend to the geometric setting, as the VC-dimension of such a range
space may again depend on the complexity of the polygon. Therefore, we will use
a different approach in Section~\ref{sec:implicit_cutting}.


\section{Computing implicit representations in subpolygon $P_r$}
\label{sec:computing_implict_representations}

Consider a diagonal $d$ that splits the polygon $P$ into two subpolygons
$P_\ell$ and $P_r$, and assume without loss of generality that $d$ is vertical
and that $P_r$ lies right of $d$. We consider only the sites $S_\ell$ in
$P_\ell$, and we restrict their functions
$F=\{f_s \cap (P_r \times \R) \mid s \in S_\ell\}$ to $P_r$. We now present a
more efficient algorithm to compute the implicit representation of $L_k(F)$ in
this setting. To this end, we show that the two-point shortest path query data
structure of Guibas and Hershberger~\cite{guibas1989query} essentially gives us
an efficient way of accessing the bisector $b_{st}$ between a pair of sites
without explicitly computing it. See Appendix~\ref{app:Bisector}. In
Section~\ref{sub:implicit_voronoi} we first use this to compute an implicit
representation of the Voronoi diagram of $S_\ell$ in $P_r$. Building on these
results, we can compute an implicit representation of the $k^\mathrm{th}$-order
Voronoi diagram $\VD_k(S_\ell)$ in $P_r$
(Section~\ref{sub:Computing_an_Implicit_k-order_Voronoi_Diagram}), the
$k$-level $L_k(F)$ of $F$ in $P_r \times \R$
(Section~\ref{sub:Computing_an_Implicit_Vertical_Decomposition}), and finally
an implicit vertical decomposition $\underline{L_k^\nabla(F)}$ of
$\underline{L_k(F)}$ in $P_r$
(Section~\ref{sub:Computing_an_Implicit_Vertical_Decomposition}).

\subsection{Computing an implicit Voronoi diagram}
\label{sub:implicit_voronoi}

The Voronoi diagram $\VD=\VD(S_\ell)$ in $P_r$ is a
forest~\cite{aronov1989geodesic}. We now show that we can compute (the topology
of) this forest efficiently by considering it as an abstract Voronoi
diagram~\cite{klein1993avd}. Our forest stores only the locations of the degree
one and degree three vertices and their adjacencies. This turns out to be
sufficient to still answer point location queries efficiently.

Assuming that certain geometric primitives like computing the intersections
between ``related'' bisectors take $O(X)$ time we can construct an abstract
Voronoi diagram of $n$ sites in expected $O(Xn\log n)$
time~\cite{klein1993avd}. We will show that \VD is a actually a
\emph{Hamiltonian} abstract voronoi diagram, which means that it can be
constructed in $O(Xn)$ time~\cite{klein1994hamiltonian_vd}. We show this in
Section~\ref{sub:Hamiltonian_avd}. In Section~\ref{sub:Geometric_Primitives} we
discuss the geometric primitives used by the algorithm of Klein and
Lingas~\cite{klein1994hamiltonian_vd}; essentially computing (a representation
of) the concrete Voronoi diagram of five sites. We show that we can implement
these primitives in $O(\log^2 m)$ time by computing the intersection point
between two ``related'' bisectors $\br_{st}$ and $\br_{tu}$. This then gives us
an $O(n\log^2 m)$ time algorithm for constructing \VD. Finally, in
Section~\ref{sub:query} we argue that having only the topological structure \VD
is sufficient to find the site in $S_\ell$ closest to a query point
$q \in P_r$.

\subsubsection{Hamiltonian abstract Voronoi diagrams}
\label{sub:Hamiltonian_avd}

In this section we show that we can consider $\VD$ as a Hamiltonian abstract
Voronoi diagram. A Voronoi diagram is \emph{Hamiltonian} if there is a curve
--in our case the diagonal $d$-- that intersects all regions exactly once, and
furthermore this holds for all subsets of the
sites~\cite{klein1994hamiltonian_vd}. Let $S_\ell$ be the set of sites in $P_\ell$
that we consider, and let $T_\ell$ be the subset of sites from $S_\ell$ whose Voronoi
regions intersect $d$, and thus occur in \VD.

\begin{lemma}
  \label{lem:hamiltonian_avd}
  The Voronoi diagram $\VD(T_\ell)$ in $P_r$ is a Hamiltonian abstract Voronoi diagram.
\end{lemma}

\begin{proof}
  By Lemma~\ref{lem:bst_intersect_d} any bisector $b_{st}$ intersects the
  diagonal $d$ at most once. This implies that for any subset of sites
  $T \subseteq S_\ell$, so in particular for $T_\ell$, the diagonal $d$ intersects
  all Voronoi regions in $\VD(T)$ at most once. By definition, $d$ intersects
  all Voronoi regions of the sites in $T_\ell$ at least once. What remains is to
  show that this holds for any subset of $T_\ell$.  This follows since the Voronoi
  region $V(s,T_1 \cup T_2)$ of a site $s$ with respect to a set $T_1 \cup T_2$
  is contained in the voronoi region $V(s,T_1)$ of $s$ with respect to $T_1$.
\end{proof}

\subparagraph{Computing the order along $d$.} We will use the algorithm of Klein %
and Lingas~\cite{klein1994hamiltonian_vd} to construct
$\VD=\VD(S_\ell)=\VD(T_\ell)$. To this end, we need the set of sites $T_\ell$ whose
Voronoi regions intersect $d$, and the order in which they do so. Next, we show
that we can maintain the sites in $S_\ell$ so that we can compute this information
in $O(n\log^2 m)$ time.


\begin{lemma}
  \label{lem:order}
  Let $s_1,..,s_n$ denote the sites in $S_\ell$ ordered by increasing distance from the
  bottom-endpoint $p$ of $d$, and let $t_1,..,t_z$ be the subset $T_\ell
  \subseteq S_\ell$ of sites whose Voronoi regions intersect $d$, ordered along $d$
  from bottom to top. For any pair of sites $t_a=s_i$ and $t_c=s_j$, with
  $a < c$, we have that $i < j$.
\end{lemma}

\begin{proof}
  Since $t_a$ and $t_c$ both contribute Voronoi regions intersecting $d$,
  their bisector must intersect $d$ in some point $w$ in between these two
  regions. Since $a < c$ it then follows that all points on $d$ below $w$,
  so in particular the bottom endpoint $p$, are closer to $t_a=s_i$ than to
  $t_c = s_j$. Thus, $i < j$.
\end{proof}

Lemma~\ref{lem:order} suggests a simple iterative algorithm for
extracting $T_\ell$ from $S_\ell=s_1,..,s_n$.

\begin{lemma}
  \label{lem:compute_sequence}
  Given $S_\ell=s_1,..,s_n$, $T_\ell$ can be computed from $S_\ell$ in $O(n\log^2 m)$
  time.
\end{lemma}




\begin{proof}
  We consider the sites in $S_\ell$ in increasing order, while maintaining $T_\ell$
  as a stack. More specifically, we maintain the invariant that when we start
  to process $s_{j+1}$, $T_\ell$ contains exactly those sites among $s_1,..,s_j$
  whose Voronoi region intersects $d$, in order along $d$ from bottom to top.

  Let $s_{j+1}$ be the next site that we consider, and let $t=s_i$, for some
  $i \leq j$, be the site currently at the top of the stack. We now compute the
  distance $\geodlen(s_{j+1},q)$ between $s_{j+1}$ and the topmost endpoint $q$ of
  $d$. If this distance is larger than $\geodlen(t,q)$, it follows that the
  Voronoi region of $s_{j+1}$ does not intersect $d$: since the bottom endpoint
  $p$ of $d$ is also closer to $t=s_i$ than to $s_{j+1}$, all points on $d$ are
  closer to $t$ than to $s_{j+1}$.

  If $\geodlen(s_{j+1},q)$ is at most $\geodlen(t,q)$ then the Voronoi region of
  $s_{j+1}$ intersects $d$ (since $t$ was the site among $s_1,..,s_j$ that was
  closest to $q$ before). Furthermore, since $p$ is closer to $t=s_i$ than to
  $s_{j+1}$ the bisector between $s_{j+1}$ and $t$ must intersect $d$ in some
  point $a$. If this point $a$ lies above the intersection point $c$ of $d$
  with the bisector between $t$ and the second site $t'$ on the stack, we have
  found a new additional site whose Voronoi region intersects $d$. We push
  $s_{j+1}$ onto $T_\ell$ and continue with the next site $s_{j+2}$. Note that the
  Voronoi region of every site intersects $d$ in a single segment, and thus
  $T_\ell$ correctly represents all sites intersecting $d$. If $a$ lies below $c$
  then the Voronoi region of $t$, with respect to $s_1,..s_{j+1}$, does not
  intersect $d$. We thus pop $t$ from $T_\ell$, and repeat the above procedure,
  now with $t'$ at the top of the stack.

  Since every site is added to and deleted from $T_\ell$ at most once the
  algorithm takes a total of $O(n)$ steps. Computing $\geodlen(s_{j+1},q)$ takes
  $O(\log m)$ time, and finding the intersection between $d$ and the bisector
  of $s_{j+1}$ and $t$ takes $O(\log^2 m)$ time
  (Lemma~\ref{lem:finding_w}). The lemma follows.
\end{proof}

We now simply maintain the sites in $S_\ell$ in a balanced binary search tree on
increasing distance to the bottom endpoint $p$ of $d$. It is easy to maintain
this order in $O(\log m + \log k)$ time per upate. We then extract the set of sites $T_\ell$
that have a Voronoi region intersecting $d$, and thus $P_r$, ordered along $d$
using the algorithm from Lemma~\ref{lem:compute_sequence}.

\subsubsection{Implementing the required geometric primitives}
\label{sub:Geometric_Primitives}

\begin{wrapfigure}[26]{r}{0.29\textwidth}
  \centering
  \vspace{-.5\baselineskip}
  \includegraphics{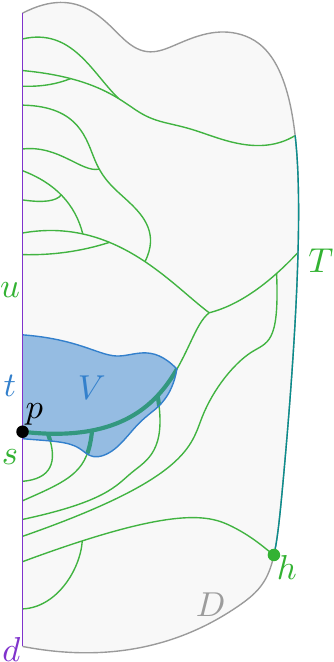}
  \caption{The part of $T$, the tree representing the Hamiltonian Voronoi
    diagram (green), that lies inside the Voronoi region $V$ (blue) of a new
    site $t$ is a subtree $T'$ (fat). We can compute $T'$, by exploring $T$
    from a point $p$ inside $V$. In case $t$ is the first site in the ordering
    along $d$ we can start from the root $h$ of the ``first'' tree in \VD.
  }
  \label{fig:insert_avd}
\end{wrapfigure}
In this section we discuss how to implement the geometric primitives needed by
the algorithm of Klein and Lingas~\cite{klein1994hamiltonian_vd}. 
They describe their algorithm in terms of the following two basic operations:
(i) compute the concrete Voronoi diagram of five sites, and (ii) insert a new
site $s$ into the existing Voronoi diagram $\VD(S)$. In their analysis, this
first operation takes constant time, and the second operation takes time
proportional to the size of $\VD(S)$ that lies inside the Voronoi region of $t$
in $\VD(S\cup\{t\})$. We observe that that to implement these operations it is
sufficient to be able to compute the intersection between two ``related''
bisectors $b_{st}$ and $b_{tu}$ --essentially computing the Voronoi diagram of
three sites-- and to test if a given point $q$ lies on the $s$-side of the
bisector $b_{st}$ (i.e.~testing if $q$ is ``closer to'' $s$ than to $t$). We
then show that in our setting we can implement these operations in
$O(\log^2 m)$ time, thus leading to an $O(n\log^2 m)$ time algorithm to compute
$\VD$.

\subparagraph{Inserting a new site.} Klein and
Lingas~\cite{klein1994hamiltonian_vd} sketch the following algorithm to insert
a new site $t$ into the Hamiltonian Voronoi diagram $\VD(S)$ of a set of sites
$S$. We provide some missing details of this procedure, and briefly argue that
we can use it to insert into a diagram of three sites. Let $D$ denote the
domain in which we are interested in $\VD(S)$ (in our application, $D$ is the
subpolygon $P_r$) and let $d$ be the curve that intersects all regions in
$\VD(S)$. Recall that $\VD(S)$ is a forest. We root all trees such that the
leaves correspond to the intersections of the bisectors with $d$. The roots of
the trees now corresponds to points along the boundary $\partial D$ of $D$. We
connect them into one tree $T$ using curves along $\partial D$. Now consider
the Voronoi region $V$ of $t$ with respect to $S\cup\{t\}$, and observe that
$T \cap V$ is a subtree $T'$ of $T$. Therefore, if we have a starting point $p$
on $T$ that is known to lie in $V$ (and thus in $T'$), we can compute $T'$
simply by exploring $T$. To obtain $\VD(S \cup \{t\})$ we then simply remove
$T'$, and connect up the tree appropriately. See Fig.~\ref{fig:insert_avd} for
an illustration. We can test if a vertex $v$ of $T$ is part of $T'$ simply by
testing if $v$ lies on the $t$-side of the bisector between $t$ and one of
the sites defining $v$. We can find the exact point $q$ where an edge
$(u,v)$ of $T$, representing a piece of a bisector $b_{su}$ leaves $V$ by
computing the intersection point of $b_{su}$ with $b_{tu}$ and $b_{st}$.

We can find the starting point $p$ by considering the order of the Voronoi
regions along $d$. Let $s$ and $u$ be the predecessor and successor of $t$ in
this order. Then the intersection point of $d$ with $b_{su}$ must lie in
$V$. This point corresponds to a leaf in $T$. In case $t$ is the first site in
the ordering along $d$ we start from the root $h$ of the tree that contains the
bisector between the first two sites in the ordering; if this point is not on
the $t$-side of the bisector between $t$ and one of the sites defining $h$ then
$b_{tu}$ forms its own tree (which we then connect to the global root). We do
the same when $t$ is the last point in the ordering. This procedure requires
$O(|T'|)$ time in total (excluding the time it takes to find $t$ in the
ordering of $S$; we already have this information when the procedure is used in
the algorithm of Klein and Lingas~\cite{klein1994hamiltonian_vd}).

We use the above procedure to compute the Voronoi diagram of five sites in
constant time: simply pick three of the sites $s$, $t$, and $u$, ordered along
$d$, compute their Voronoi diagram by computing the intersection of $b_{st}$
and $b_{tu}$ (if it exists), and insert the remaining two sites. Since the
intermediate Voronoi diagrams have constant size, this takes constant time.

\subparagraph{Computing the intersection of bisectors $\br_{st}$ and $\br_{tu}$. }
Since \VD is a Hamiltonian Voronoi diagram, any any pair of bisectors
$\br_{st}$ and $\br_{tu}$, with $s,t,u \in T_\ell$, intersect at most once
(Lemma~\ref{lem:intersect_once}). Next, we show how to compute
this intersection point (if it exists).



\begin{lemma}
  \label{lem:binary_search_funnel}
  Given $\hat{\P}(z,s,t)$ and $\hat{\P}(z',t,u)$, finding the
  intersection point $p$ of $\br_{st}$ and $\br_{tu}$ (if it exists) takes
  $O(\log^2 m)$ time.
\end{lemma}

\begin{wrapfigure}[12]{r}{0.47\textwidth}
  \centering
  \includegraphics{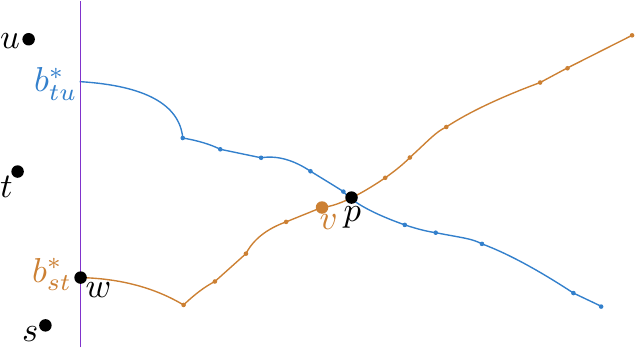}
  \caption{We find the intersection point of the two bisectors by binary
    searching along $\br_{st}$. }
  \label{fig:find_intersection_binsearch}
\end{wrapfigure}
\proofcmd
  We will find the edge of $\br_{st}$ containing the intersection point $p$ by
  binary searching along the vertices of $\br_{st}$. Analogously we find the
  edge of $\br_{tu}$ containing $p$. It is then easy to compute the exact
  location of $p$ in constant time.

  Let $w$ be the starting point of $\br_{st}$, i.e.~the intersection of
  $b_{st}$ with $d$, and assume that $t$ is closer to $w$ than $u$, that is,
  $\geodlen(t,w) < \geodlen(u,w)$ (the other case is symmetric). In our binary
  search, we now simply find the last vertex $v=v_k$ for which
  $\geodlen(t,v) < \geodlen(u,v)$. It then follows that $p$ lies on the edge
  $(v_k,v_{k+1})$ of $\br_{st}$. See
  Fig.~\ref{fig:find_intersection_binsearch}. Using
  Lemma~\ref{lem:random_access_bst_final} we can access any vertex of
  $\br_{st}$ in $O(\log m)$ time. Thus, this procedure takes $O(\log^2 m)$ time
  in total.  \hfill\qed

Note that we can easily extend the algorithm from
Lemma~\ref{lem:binary_search_funnel} to also return the actual edges of
$\br_{st}$ and $\br_{tu}$ that intersect. With this information we can
construct the cyclic order of the edges incident to the vertex of \VD
representing this intersection point. It now follows that for every group $S_\ell$
of sites in $P_\ell$, we can compute a representation of \VD of size $O(k)$ in
$O(k\log^2 m)$ time.



\subsubsection{Planar point location in \VD}
\label{sub:query}

In this section we show that we can efficiently answer point location queries,
and thus nearest neighbor queries using \VD.


\begin{lemma}
  \label{lem:bisectors_x-monotone}
  For $s,t \in S_\ell$, the part of the bisector $\br_{st} = b_{st} \cap P_r$
  that lies in $P_r$ is $x$-monotone.
\end{lemma}

  \begin{wrapfigure}[15]{r}{0.31\textwidth}
    \centering
    \includegraphics{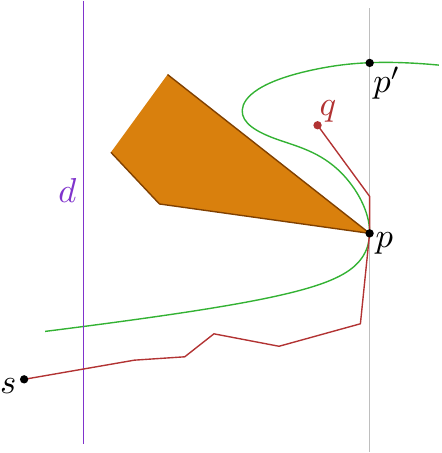}
    \caption{A non $x$-monotone bisector can occur only in \mbox{degenerate} inputs.}
    \label{fig:non_x-monotone}
  \end{wrapfigure}

\proofcmd
  Assume, by contradiction, that $b_{st}$ is not $x$-monotone in $P_r$, and let
  $p$ be a point on $b_{st}$ such that $p_x$ is a local maximum. Since
  $\br_{st}$ is not $x$-monotone, it intersects the vertical line through $p$
  also in another point $p'$ further along $\br_{st}$. Let $q$ be a point in
  the region enclosed by the subcurve along $\br_{st}$ from $p$ to $p'$ and
  $\overline{pp'}$. See Fig.~\ref{fig:non_x-monotone}. This means that either
  $\geod(s,q)$ or $\geod(t,q)$ is non $x$-monotone. Assume without loss of
  generality that it is $\geod(s,q)$. It is now easy to show that $\geod(s,q)$
  must pass through $p$. However, that means that $\br_{st}$ (and thus
  $b_{st}$) touches the polygon boundary in $p$. By the general position
  assumption $b_{st}$ has no points in common with $\partial P$ other than its
  end points. Contradiction.
\hfill\qed

Since the (restriction of the) bisectors are $x$-monotone
(Lemma~\ref{lem:bisectors_x-monotone}) we can preprocess \VD for point location
using the data structure of Edelsbrunner and
Stolfi~\cite{edelsbrunner_optimal_1986}. Given the combinatorial embedding of
\VD, this takes $O(|\VD|)$ time. To decide if a query point $q$ lies above or
below an edge $e \in \VD$ we simply compute the distances $\geodlen(s,q)$ and
$\geodlen(t,q)$ between $q$ and the sites $s$ and $t$ defining the bisector
corresponding to edge $e$. This takes $O(\log m)$ time. Point $q$ lies on the
side of the site that has the shorter distance. It follows that we can
preprocess \VD in $O(k)$ time, and locate the Voronoi region containing a query
point $q$ in $O(\log k\log m)$ time. We summarize our results in the following Lemma.

\begin{lemma}
  \label{lem:reconstruct_v}
  Given a set of $n$ sites $S_\ell$ in $P_\ell$, ordered by increasing distance
  from the bottom-endpoint of $d$, the forest \VD representing the Voronoi
  diagram of $S_\ell$ in $P_r$ can be computed in $O(n\log^2 m)$ time. Given
  \VD, finding the site $s \in S_\ell$ closest to a query point $q \in P_r$
  requires $O(\log n\log m)$ time.
\end{lemma}

\subsection{Computing an implicit $k^\mathrm{th}$-order Voronoi diagram}
\label{sub:Computing_an_Implicit_k-order_Voronoi_Diagram}

Based on the relation between the $i^\mathrm{th}$-order Voronoi diagram and the
$(i+1)^\mathrm{th}$-order Voronoi diagram (see
Section~\ref{sec:implicit_representations}) Lee developed an iterative
algorithm to compute the Euclidean $k^\mathrm{th}$-order Voronoi diagram in
$O(nk^2)$ time. His algorithm extends to any distance metric. Since the
geodesic distance is a metric, this approach, together with our algorithm from
the previous section gives us a way to compute $\VD_k(S_\ell)$ in $P_r$. We
obtain a $O(k^2n(\log n + \log^2 m))$ time algorithm. This then results in an
$O(k^2n(\log n + \log^2 m))$ time algorithm for computing a decomposition of
(space below the) $k$-level into pseudo-prisms.

\begin{theorem}
  \label{thm:compute_kth_order_vd}
  An implicit representation of the $k$-th order Voronoi
  diagram of $S_\ell$ in $P_r$ can be constructed in $O(k^2n(\log n + \log^2 m))$ time.
\end{theorem}

\begin{proof}
  We use Lee's algorithm to iteratively build the $k$-th order Voronoi diagram
  $\VD_k(S_\ell)$ in $P_r$. Consider a cell $C=V_i(H,S_\ell)$ in the
  $i^\mathrm{th}$-order Voronoi diagram. We collect the set of sites $Q$
  neighboring $C$ by traversing $\partial C$, and order them on increasing
  distance to the bottom endpoint of $d$. Note that each edge $e_j$ on
  $\partial C$ corresponds to a piece of bisector $b_{st}$, and the site
  $q_j \in Q$ that we are looking for is either $s$ or $t$, depending on which
  side of $b_{st}$ our cell $C$ lies. This means that we can collect the $q$
  sites in $Q$, and order them on increasing distance to the bottom endpoint of
  $d$ in $O(q(\log q + \log m))$ time. We then construct (an implicit
  representation) of its voronoi diagram $\VD(Q)$ in $P_r$ in $O(q\log^2 m)$
  time (Lemma~\ref{lem:reconstruct_v}). Finally, we clip $\VD(Q)$ to $C$. Since
  all intersection points of $\VD(Q)$ with $C$ are vertices of $\VD_i(S_\ell)$
  on $\partial C$ (see Lee~\cite{lee1982kthorder_vd}), all that remains is to
  find these points in $\VD(Q)$. We use a point location query to find one of
  the vertices of $\partial C$ in $\VD(Q)$, and then find the remaining
  vertices in $C$ using a breadth first search in $\VD(Q)$. It follows that in
  total we spend at most $O(q(\log n + \log^2 m))$ time. Summing over all
  $O(i(n-i))$ cells in $\VD_i(S_\ell)$, and all $k$ rounds, gives us a running
  time of $O(k^2n(\log n + \log^2 m)$) as claimed.
\end{proof}

Similar to in Theorem~\ref{thm:compute_kth_order_vd} we can compute the
downward projection $\underline{L_k(F)}$ of the $k$-level.

\subsection{Computing an implicit vertical decomposition of $\underline{L_k(F)}$}
\label{sub:Computing_an_Implicit_Vertical_Decomposition}

We now show how to turn our implicit representation of
$\underline{L_k(F_\ell)}$ into an implicit vertical decomposition as
follows. We note that this same procedure applies for computing a vertical
decomposition of $\VD_k(S_\ell)$.


\begin{lemma}
  \label{lem:compute_k-level}
  A representation $\underline{L_k^\nabla(F)}$ of the $k$-level $L_k(F)$
  consisting of $O(k(n-k))$ pseudo-trapezoids can be computed in
  $O(k^2n(\log n + \log^2 m))$ time. Given a query point $q \in P_r$, the
  $k$-nearest site in $S_\ell$ can be reported in $O(\log n\log m)$ time.
\end{lemma}

\begin{proof}
  For every vertex $v$, we know the faces of ${\underline{L_k(F)}}$
  incident to $v$ directly above and below $v$. The upward and downward
  extension segments will be contained in these faces, respectively. For each
  face $X$, we collect the vertices whose upward extension segment will be
  contained in $X$, and use a simple sweep line algorithm to compute which edge
  of (the implicit representation of) $X$ each such extension segment hits. For
  each vertex $v$ we then know that the upper endpoint of its upward extension
  segment lies on a bisector $b_{st}$, for some $s,t \in S_\ell$. To find the
  exact location of this endpoint, we use a binary search along $b_{st}$. We
  use the same approach for finding the bottom endpoint of the downward
  extension segment. Finding the edges of the implicit representation of
  $\VD_k(S_\ell)$ hit, takes $O(|X|\log|X|)$ time per face, summing over all
  faces this solves to $O(kn\log n)$. The final binary search to find the exact
  location of the endpoint takes $O(\log^2 m)$ time per point
  (Theorem~\ref{thm:represent_bisector}). It follows that we spend
  $O(k(n-k)(\log n + \log^2 m))$ time to compute all extension segments. This
  is dominated by the time it takes to compute $\underline{L_k(F)}$
  itself. Note that in the resulting subdivision, all faces are again monotone
  (ignoring the boundary of $P$), so we can preprocess it for efficient point
  location as in Section~\ref{sub:implicit_voronoi}.
\end{proof}

\section{An implicit shallow cutting of the geodesic distance function}
\label{sec:implicit_cutting}

Let $F$ again denote the set of geodesic distance functions that the sites
$S_\ell$ in $P_\ell$ induce in $P_r$. We now argue that we can compute an
implicit $k$-shallow cutting $\Lambda_k(F)$ for these functions.

As in Section~\ref{sec:Approximating_the_k-Level}, let $R$ be our random sample
of size $r$, and let $L_t(R)$ be our approximate $k$-level of $\A(F)$. Let
$\underline{L_t^\nabla}$ be the vertical decomposition of $L_t$. We now raise every
pseudo-trapezoid in $\underline{L_t^\nabla}$ to the $t$-level. Denote the result by
$\Lambda$. Let $F_p=F_{\rho(p)}$ denote the conflict list of $p \in \R^3$, i.e., the
functions intersecting the vertical downward half-line $\rho(p)$ starting in
$p$.

\begin{lemma}
  \label{lem:conflict_lists}
  Let $\nabla$ be a pseudo prism in $\Lambda$. The conflict list $F_\nabla$ of
  $\nabla$ is the union of the conflict lists of its corners $W$,
  i.e. $F_\nabla = \bigcup_{v \in W} F_v$.
\end{lemma}

\begin{wrapfigure}[13]{r}{0.39\textwidth}
  \centering
  \vspace{.5\baselineskip}
  \includegraphics{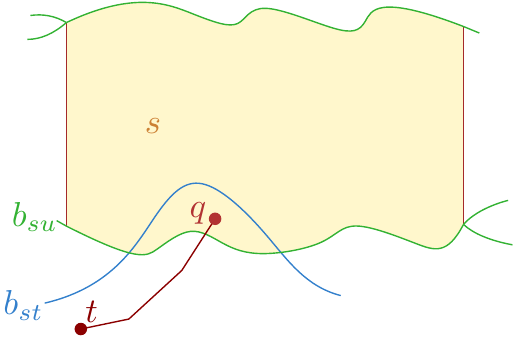}
  \caption{Since the bisectors restricted to $P_r$ are $x$-monotone it
    follows that if a site $t$ conflicts with a prism $\nabla$, it must
    conflict with a corner of $\nabla$.}
  \label{fig:conflict_list_prism}
\end{wrapfigure}
\proofcmd
  Let $f_s$ be the function defining the ceiling of $\nabla$. We have that
  $F'=\bigcup_{v \in W} F_v \subseteq F_\nabla$ by definition, so we focus on
  proving $F_\nabla \subseteq F'$. Assume by contradiction that
  $f_t \in F_\nabla$, but $f_t \not\in F'$. So, there is a point
  $q \in \underline{\nabla}$ for which $\geodlen(t,q) < \geodlen(s,q)$, but
  $\geodlen(s,v) < \geodlen(t,v)$ for all corners $v \in W$. Hence, all four corners
  lie on the ``$s$-side'' of $\br_{st}$, whereas $p$ lies on the ``$t$-side''
  of $\br_{st}$. Assume without loss of generality that $s$ is closer to the
  points above $\br_{st}$ (and thus all corners lie above $\br_{st}$). See
  Fig.~\ref{fig:conflict_list_prism}. Since $\br_{st}$ is $x$-monotone
  (Lemma~\ref{lem:bisectors_x-monotone}) it must intersect the bottom edge of
  $\underline{\nabla}$ twice. This bottom edge is part of a single bisector
  $\br_{su}$, for some $f_u \in F$. However, by
  Lemma~\ref{lem:intersect_once} $\br_{st}$ and $\br_{su}$ intersect at most
  Once. Contradiction.
\hfill\qed

\begin{theorem}
  \label{thm:shallow_cutting}
  $\Lambda$ is a vertical $k$-shallow $(k(1+\eps)/n)$-cutting of $\A(F)$ whose
  topological complexity, and thus its size, is
  $O((n/k\eps^5)\log^2 n)$. Each pseudo-prism in $\Lambda$ intersects at least
  $k$ and at most $4k(1+\eps)$ functions in $F$.
\end{theorem}

\begin{proof}
  By Lemma~\ref{lem:complexity_level} $\Lambda$ consists of
  $O((n/k\eps^5)\log^2 n)$ regions. Note that all regions are
  pseudo-prisms. Lemma~\ref{lem:conflict_lists} then gives us that the conflict
  list of each pseudo-prism is contained in the conflict lists of its at most
  four corners.
\end{proof}

\subsection{Computing the conflict lists}
\label{sub:conflict_lists}

Using Lemma~\ref{lem:compute_k-level} we can construct an implicit
representation of the $k$-shallow cutting $\Lambda = \Lambda_k(F)$. So, all
that remains is to compute the conflict lists of the pseudo-prisms. By
Lemma~\ref{lem:conflict_lists} it is sufficient to compute the conflict lists
of the four corner points of each pseudo-prism. Next, we show how to do this in
$O(n(\log^3 n\log m+\log^2 m))$ expected time.

We use the same approach as used by Chan~\cite{chan2000rangerep}. That is, we
first build a data structure on our set of functions $F$ so that for a vertical
query line $\ell$ (in $\R^3$) and a value $k$, we can report the lowest $k$
functions intersected by $\ell$ in $O((\log n + k)\log m)$ expected time. We
then extend this to report only the functions that pass strictly below some
point $q \in \R^3$. To compute the conflict lists of all corners in
$\Lambda_k(F)$ we repeatedly query this data structure.

\subparagraph{The data structure.} Our data structure consists of a hierarchy
of the lower envelopes of random samples
$R_0 \subset R_1 \subset .. \subset R_{\log n}$, where $|R_i|=2^i$. For each
set $R_i$ we store an implicit vertical decomposition representing the (the
downward projection of the) lower envelope $L_{0,i} = L_0(R_i)$. This
decomposes the space below $L_{0,i}$ into pseudo-prisms. For each such
pseudo-prism $\nabla$ we store its conflict list $F_\nabla$ with respect to
$F$, i.e.~the functions from (the entire set) $F$ that intersect $\nabla$. The
following lemma shows that for each $R_i$, the expected amount of space used is
$O(n)$. The total expected space used is thus $O(n\log n)$.

\begin{lemma}
  \label{lem:lower_env_clarkson_shor}
  Let $r \in [1,n]$ and consider a random sample $R$ of $F$ of size $r$. (i)
  The expected value of $\sum_\nabla |F_\nabla|$ over all pseudo-prisms below
  $L_0(R)$ is $O(n)$, and (ii) For any vertical line $\ell$, the expected
  value of $|F_\nabla|$, where
  $\nabla$ is the pseudo-prism of $L_0(R)$ intersected by $\ell$, is $O(n/r)$.
\end{lemma}

\begin{proof}
  The first statement follows directly from a Clarkson and Shor style sampling
  argument. More specifically, from what
  Har-Peled~\cite{harPeled2011geometric} calls the ``Bounded moments
  theorem'' (Theorem 8.8). The second statement then follows directly from the first statement.
\end{proof}

\subparagraph{Building the data structure.} For each set $R_i$, we use the
algorithm from Section~\ref{sec:computing_implict_representations} to construct
an implicit vertical decomposition of $\underline{L_{0,i}}$. To this end, we
need to order the (sites corresponding to the) functions in $R_i$ on increasing
distance to the bottom endpoint of the diagonal $d$. For $R_{\log n} = F$ we do
this in $O(n(\log n + \log m))$ time. For $R_{i-1}$ we do this by filtering the
ordered set $R_i$ in linear time. Since the sizes of $R_i$ are geometrically
decreasing, it follows that we spend $O(n(\log n + \log^2 m))$ time in
total. 


\begin{lemma}
  \label{lem:conflict}
  Let $f_s \in F \setminus R$ be a function that intersects a pseudo-prism of
  $L_0(R)$, let $T$ be the set of sites whose functions contribute to
  $L_0(R)$, ordered on increasing distance from the bottom endpoint of $d$,
  and let $t$ and $u$ be the predecessor and successor of $s$ in $T$,
  respectively. The vertex $v \in \underline{L_0(R)}$ that represents
  $d \cap b_{tu}$ is closer to $s$ than to $t$ and $u$.
\end{lemma}

\begin{proof}
  If $f_s$ intersects a pseudo prism of $L_0(R)$ then there is a point
  $q \in P_r$ for which $s$ is closer than all other sites in $T$. It follows
  that there must be a point on the diagonal $d$ that is closer to $s$ than to
  all other sites in $T$. Lemma~\ref{lem:order} then gives us that the
  Voronoi region of $s$ (with respect to $R \cup \{s\}$) on $d$ must lie in
  between that of $t$ and $u$ (if these still contribute a Voronoi
  region). Therefore, $t$ and $u$ no longer have a vertex of
  $\VD(R \cup \{s\})$ on $d$. Since $t$ and $u$ were the closest sites to $v$
  in $R$, this implies that $v$ must lie in the Voronoi region of $s$, hence
  $s$ is closer to $v$ than $t$ and $u$.
\end{proof}

By Lemma~\ref{lem:conflict} we can now compute the conflict lists of the cells
in $L_{0,i}$ as follows. For each function $f_s \in F\setminus R_i$ we find the
vertex $v$ defined in Lemma~\ref{lem:conflict}. If $s$ is further from $v$ than
the sites defining it, then $f_s$ does not conflict with any pseudo-prism in
$L_{0,i}$. Otherwise, we find \emph{all} (degree one or degree three) vertices
of $\underline{L_{0,i}}$ that conflict with $s$. Since Voronoi regions are
simply connected, we can do this using a breadth first search in
$\underline{L_{0,i}}$, starting from vertex $v$. When we have this information
for all functions in $F \setminus R_i$, we actually also know for every vertex
$v$ in $L_{0,i}$ which functions $F \setminus R_i$ pass below it. That is, we
have the conflict lists for all vertices $v$. The conflict list of a
pseudo-prism in $L_{0,i}$ is then simply the union of the conflict lists of its
four corners (Lemma~\ref{lem:conflict_lists}).

Given the ordering of all sites in $S$ on increasing distance to the bottom
endpoint of $d$, we can find the initial vertices for all functions in
$F \setminus R_i$ in $O(|R_i|\log m))$ time. For every other reported conflict we
spend $O(\log m)$ time, and thus computing the conflict lists for all cells in
$L_{0,i}$ takes $O(\sum_{\nabla \in L_{0,i}} |F_\nabla|\log m)$ time. By
Lemma~\ref{lem:lower_env_clarkson_shor} this sums to $O(n\log m)$ in
expectation. Summing over all $O(\log n)$ random samples, it follows that we
spend $O(n\log n\log m)$ expected time to compute all conflict lists. The total
expected time to build the data structure is thus $O(n(\log^2 m + \log n\log
m))$.

\subparagraph{Querying.} The query algorithm is exactly as in
Chan~\cite{chan2000rangerep}. The main idea is to use a query algorithm that
may fail, depending on some parameter $\delta$, and then query with varying
values of $\delta$ until it succeeds. The query algorithm locates the cell
$\nabla$ in $L_0(R_i)$ stabbed by the vertical line $\ell$, for
$i=\lceil \log\lceil n\delta/k\rceil\rceil$. If $|F_\nabla| > k/\delta^2$ or
$|F_\nabla \cap \ell| < k$ the query algorithm simply fails. Otherwise it
reports the $k$ lowest functions intersecting $\ell$. Since computing the
intersection of a function $f_s$ with $\ell$ takes $O(\log m)$ time, the
running time is $O((\log n + k/\delta^2)\log m)$. Using three independent
copies of the data structure, and querying with $\delta = 2^{-j}$ for
increasing $j$ gives us an algorithm that always succeeds in
$O((\log n + k)\log m)$ time. Refer to Chan~\cite{chan2000rangerep} for details. We
can now also report all functions that pass below a point $q$ by repeatedly
querying with the vertical line through $q$ and doubling the value of $k$. This
leads to a query time of $O((\log n + k)\log m)$, where $k$ is the number of
functions passing below $q$.


\begin{theorem}
  \label{thm:k_nn_data_structure}
  There is a data structure of size $O(n\log n)$ that allows reporting the $k$
  lowest functions in $\A(F)$ intersected by a vertical line through a query
  point $q \in P_r$, that is, the $k$-nearest neighbors of a query point $q$,
  or all $k$ functions that pass below $q$, in $O((\log n + k)\log m)$
  time. Building the data structure takes $O(n(\log n \log m + \log^2 m))$
  expected time.
\end{theorem}

\subparagraph{Computing a shallow cutting.} To construct a shallow cutting we
now take a random sample $R$ of size $r$, build an implicit representation of
the $t$-level in this sample, and then construct the above data structure to
compute the conflict lists. By Lemma~\ref{lem:compute_k-level} constructing the
implicit representation of $L_t(R)$ takes $O(t^2r(\log r + \log^2 m))$
time. Plugging in $r=(cn/k\eps^2)\log n$, $t=\Theta(1/\eps^2\log n)$, and
$\eps = 1/2$, this takes $O((n/k)\log^3 n(\log n + \log^2 m))$ expected time.

Constructing the query data structure takes $O(n(\log n\log m + \log^2 m))$
time. We then query it with all degree three and degree one vertices in
$\Lambda$. The total size of these conflict lists is $O(n\log^2 n)$
(Theorem~\ref{thm:shallow_cutting}). So, this takes $O(n\log^3 n\log m)$ time
in total. We conclude:

\begin{theorem}
  \label{thm:compute_shallow_cutting}
  A $k$-shallow cutting $\Lambda_k(F)$ of $F$ of topological
  complexity $O((n/k)\log^2 n)$ can be computed in
  $O((n/k)\log^3 n(\log n + \log^2 m) + n\log^2 m + n\log^3 n\log m)$ expected
  time.
\end{theorem}

\section{Putting everything together}
\label{sec:combining}


Kaplan \etal~\cite{kaplan2017dynamic} essentially prove the following result,
which, combined with Theorem~\ref{thm:compute_shallow_cutting} gives us an
efficient data structure to answer nearest neighbor queries when sites are in
$P_\ell$ and the query points are in $P_r$.

\begin{lemma}[Kaplan \etal~\cite{kaplan2017dynamic}]
  \label{lem:kaplan_DS}
  Given an algorithm to construct a $k$-shallow cutting $\Lambda$ of size
  $S(n,k)$ on $n$ functions in $T(n,k)$ time, and such that locating the
  cell $\nabla$ in $\Lambda$ containing a query point $q$ takes $Q(n,k)$ time, we
  can construct a data structure of size $O(S(n,k)\log n)$ that maintains a
  dynamic set of at most $n$ functions $F$ and can report the function that
  realizes the lower envelope $L_0(F)$ at a query point $q$ in
  $O(Q(n,1)\log n)$ time. Inserting a new function in $F$ takes
  $O((T(n,1)/n)\log n)$ amortized time, and deleting a function from $F$ takes
  $O((T(n,1)/n)\log^3 n)$ amortized time.
\end{lemma}

Our main data structure is a balanced binary tree, corresponding to a balanced
decomposition of $P$ into sub-polygons~\cite{guibas1987balanced}, in which each
node stores two copies of the data structure from Lemma~\ref{lem:kaplan_DS}. A
node in the tree corresponds to a subpolygon $P'$ of $P$, and a diagonal $d$
that splits $P'$ into two roughly equal size subpolygons $P_\ell$ and
$P_r$. One copy of our data structure associated with this node stores the
sites in $S_\ell$ and can answer queries in $P_r$. The other copy stores the
sites in $S_r$ and can answer queries in $P_\ell$. Since the balanced
hierarchical decomposition consists of $O(\log m)$ layers, every site is stored
$O(\log m)$ times. This results in an $O(n\log^3n\log m + m)$ size data
structure. To answer a query $q$, we query $O(\log m)$ data structures, one at
every level of the tree, and we report the site that is closest over all. 

\begin{repeattheorem}{thm:fully_dynamic_ds_polylog}
  Let $P$ be a simple polygon $P$ with $m$ vertices. There is a fully dynamic data
  structure of size $O(n\log^3 n\log m + m)$ that maintains a set of
  $n$ point sites in $P$ and allows for geodesic nearest neighbor queries in
  worst case $O(\log^2 n\log^2 m)$ time. Inserting a site takes
  $O(\log^5 n\log m + \log^4 n\log^3 m)$ amortized expected time, and deleting
  a site takes $O(\log^7n\log m + \log^6 n\log^3 m)$ amortized expected time.
\end{repeattheorem}

\begin{proof}
  Theorem~\ref{thm:compute_shallow_cutting} gives us
  $T(n,k) = O((n/k)\log^3 n(\log n + \log^2 m) + n\log^2 m + n\log^3 n\log m)$,
  $S(n,k) = O(n\log^2 n)$, and $Q(n,k) = O(\log n\log m)$, where $m$ is the
  size of our polygon. Therefore, $T(n,1)/n = O(\log^4 n + \log^3 n\log^2
  m)$. Plugging in these results in Lemma~\ref{lem:kaplan_DS} and using that
  the balanced decomposition consists of $O(\log m)$ levels completes the proof.
\end{proof}

\section{An Improved Data Structure for Offline-Updates or Insertions-Only}
\label{sec:simplified}

In this Section we briefly sketch how to use some of the tools and techniques
we developed to gen an alternative, simpler data structure for nearest neighbor
queries. For a fully dynamic scenario this data structure is slower than our
result from Theorem~\ref{thm:fully_dynamic_ds_polylog}, however in case there
are no deletions, or the full sequence of updates is known in advance, this
method is actually faster than the data structure of
Theorem~\ref{thm:fully_dynamic_ds_polylog}.

The main idea is still to recursively partition the polygon into a ``left''
subpolygon $P_\ell$ and a ``right'' subpolygon $P_r$. instead of building a
dynamic lower envelope data structure of the sites $S_\ell = S \cap P_\ell$ in
$P_r$ we use the following approach. We further split the sites $S_\ell$ into
subsets $S_1,..,S_k$, and for each subset we use the algorithm from
Section~\ref{sub:implicit_voronoi} to build (an implicit representation of) the
Voronoi diagram they induce in $P_r$. To answer a query (of a query point in
$P_r$) we simply query \emph{all} $k$ (implicit) Voronoi diagrams. To update
the data structure we simply rebuild the Voronoi diagram(s) of the affected
subset(s). For a simple fully dynamic data structure we can partition the sites
$S_\ell$ in $k=O(\sqrt{n})$ groups of size $O(\sqrt{n})$ each, to get
$O(\sqrt{n}\polylog n\polylog m)$ worst case update and query times.

\subparagraph{Insertions only.} In case our data structure has to support only
insertions, we can improve the insertion time to $O(\log n\log^3 m)$, albeit
being amortized. We now partition the sites in $S_\ell$ in groups of size
$2^i$, for $i \in [1..O(\log n)]$. When we insert a new site, we may get two
groups of size $2^i$. We then remove these groups, and construct a new group of
size $2^{i+1}$. For this group we rebuild the Voronoi diagram $\VD$ that these
sites induce on $P_r$ from scratch. 
Using a standard binary counter argument it can be shown that every data
structure of size $2^{i+1}$ gets rebuild (at a cost of $O(2^i\log^2 m)$) only
after $2^i$ new sites have been inserted~\cite{overmars1983design}. So, if we
charge $O(\log n\log^2 m)$ to each site, it can pay for rebuilding all of the
structures it participates in. We do this for all $O(\log m)$ levels in the
balanced decomposition, hence we obtain the following result.

\begin{theorem}
  \label{thm:insert_only}
  Let $P$ be a simple polygon $m$ vertices. There is an insertion-only data
  structure of size $O(n\log m + m)$ that stores a set of $n$ point sites in
  $P$, allows for geodesic nearest neighbor queries in worst-case
  $O(\log^2 n\log^2 m)$ time, and inserting a site in amortized
  $O(\log n\log^3 m)$ time.
\end{theorem}

\subparagraph{Offline updates.} When we have both insertions and deletions, but the
order of these operations is known in advance, we can maintain $S$ in amortized
$O(\log n\log^3 m)$ time per update, where $n$ is the maximum number of sites
in $S$ at any particular time. Queries take $O(\log^2 n\log^2 m)$ time, and may
arbitrarily interleave with the updates. Furthermore, we do not have to know
them in advance.

For ease of description, we assume that the total number of updates $N$ is
proportional to the number of sites at any particular time, i.e.~$N \in
O(n)$. We can easily extend our approach to larger $N$ by grouping the updates
in $N/n$ groups of size $O(n)$ each. Consider a node of the balanced
decomposition whose diagonal that splits its subpolygon into $P_\ell$ and
$P_r$. We partition the sites in $S_\ell$ into groups such that at any time, a
query $q \in P_r$ can be answered by considering the Voronoi diagrams in $P_r$
of only $O(\log n)$ groups. We achieve this by building a segment
tree on the intervals during which the sites are
``alive''. More specifically, let $[t_1,t_2]$ denote a time interval in which
a site $s$ should occur in $S_\ell$ (i.e.~$s$ lies in $P_\ell$ and there is an
\acall{Insert}{$s$} operation at time $t_1$ and its corresponding
\acall{Delete}{$s$} at time $t_2$). We store the intervals of all sites in
$S_\ell$ in a segment tree~\cite{bkos2008}. Each node $v$ in this tree is
associated with a subset $S_v$ of the sites from $S_\ell$. We build the Voronoi
diagram that $S_v$ induces on $P_r$. Every site occurs in $O(\log n)$ subsets,
and in $O(\log m)$ levels of the balanced decomposition, so the total size of
our data structure is $O(n\log n\log m + m)$. Building the Voronoi diagram for
each node $v$ takes $O(|S_v|\log^2 m)$ time. Summing these results over all
nodes in the tree, and all levels of the balanced decomposition, the total
construction time is $O(n\log n\log^3 m + m)$ time. We conclude:

\begin{theorem}
  \label{thm:offline}
  Let $P$ be a simple polygon $P$ with $m$ vertices, and let $\mathcal{S}$ be a
  sequence of operations that either insert a point site inside $P$ into a set
  $S$, or delete a site from $S$. There is a dynamic data structure of size
  $O(n\log n\log m + m)$, where $n$ is the maximum number of sites in $S$ at
  any time, that stores $S$, and allows for geodesic nearest neighbor queries
  in $O(\log^2 n\log^2 m)$ time. Updates take amortized $O(\log n\log^3 m)$
  time.
\end{theorem}




\bibliography{improved_geod_nn}

\newpage
\appendix

\section{Representing and Computing a Bisector}
\label{app:Bisector}

Assume without loss of generality that the diagonal $d$ that splits $P$ into
$P_\ell$ and $P_r$ is a vertical line-segment, and let $s$ and $t$ be two sites
in $S_\ell$. In this section we show that there is a representation of
$\br_{st} = b_{st} \cap P_r$, the part of the bisector $b_{st}$ that lies in
$P_r$, that allows efficient random access to the bisector vertices. Moreover,
we can obtain such a representation using a slightly modified version of the
two-point shortest path data structure of Guibas and
Hershberger~\cite{guibas1989query}.

Let $s$ be a site in $S_\ell$, and consider the shortest path tree $T$ rooted
at $s$. Let $e=\overline{uv}$ be an edge of $T$ for which $v$ is further away from $s$
than $u$. The half-line starting at $v$ that is colinear with, and extending
$e$ has its first intersection with the boundary $\partial P$ of $P$ in a point $w$. We refer to
the segment $\overline{vw}$ as the \emph{extension segment} of
$v$~\cite{aronov1989geodesic}. Let $E_s$ denote the set of all extension
segments of all vertices in $T$.

\begin{figure}[tb]
  \centering
  \includegraphics[page=2]{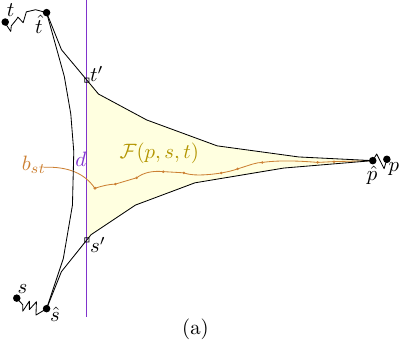}
  \quad
  \includegraphics[page=3]{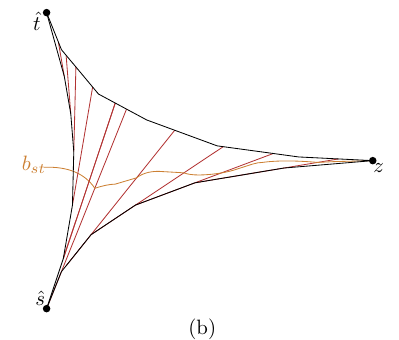}
  \caption{(a) The polygon $\P(p,s,t)$ bounded by the shortest paths between
    $s$, $p$, and $t$ is a pseudo-triangle $\hat{\P}(p,s,t)$ with polylines
    attached to its corners $\hat{s}$, $\hat{p}$, and $\hat{t}$. It contains
    the funnel $\F(p,s,t)$. (b) The clipped extension segments in $F_s^t$ are all
    pairwise disjoint, and end at the chain from $t$ to $z$.}
  \label{fig:bisector_in_funnel}
\end{figure}

Consider two sites $s,t \in S_\ell$, and its bisector $b_{st}$. We then have

\begin{lemma}[Lemma 3.22 of Aronov~\cite{aronov1989geodesic}]
  \label{lem:bisector_prop}
  The bisector $b_{st}$ is a smooth curve connecting two points on $\partial P$
  and having no other points in common with $\partial P$. It is the
  concatenation of $O(m)$ straight and hyperbolic arcs. The points along
  $b_{st}$ where adjacent pairs of these arcs meet, i.e.,~the vertices of
  $b_{st}$, are exactly the intersections of $b_{st}$ with the segments of $E_s$ or
  $E_t$.
\end{lemma}

\begin{lemma}[Lemma~3.28 of Aronov~\cite{aronov1989geodesic}]
  \label{lem:bisector_sp_intersect}
  For any point $p \in P$, the bisector $b_{st}$ intersects the shortest path
  $\geod(s,p)$ in at most a single point.
\end{lemma}

\begin{figure}[tb]
  \centering
  \includegraphics{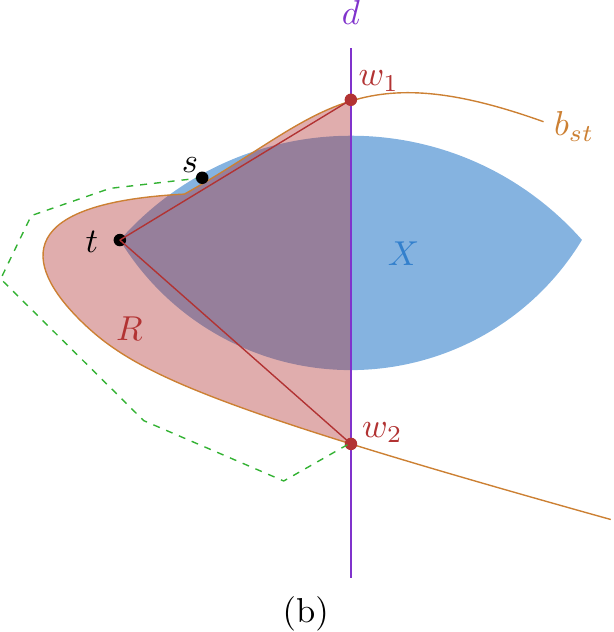}
  \caption{The
    geodesic distance from $t$ to $w_1$ and $w_2$ equals its Euclidean
    distance. The shortest path from $s$ to $w_2$ (dashed, green) has to go
    around $R$, and is thus strictly longer than $\|tw_2\|$.}
  \label{fig:intersection_diagonal}
\end{figure}


Consider a point $p$ on $\partial P_r$ and let $\P(p,s,t)$ be the polygon
defined by the shortest paths $\geod(s,p)$, $\geod(p,t)$, and
$\geod(t,s)$. This polygon $\P(p,s,t)$ is a pseudo-triangle $\hat{\P}(p,s,t)$
whose corners $\hat{s}$, $\hat{t}$, and $\hat{p}$, are connected to $s$, $t$,
and $p$ respectively, by arbitrary polylines.

Let $s'$ and $t'$ be the intersection points between $d$ and the geodesics
$\geod(p,s)$ and $\geod(p,t)$, respectively, and assume without loss of generality that $s'_y
\leq t'_y$. The restriction of $\P(p,s,t)$ to $P_r$ is a
\emph{funnel} $\F(p,s,t)$, bounded by $\geod(t',p)$, $\geod(p,s')$, and
$\overline{s't'}$. See Fig.~\ref{fig:bisector_in_funnel}(a). Note that 
$\geod(s,t)$ is contained in $P_\ell$.

Clearly, if $b_{st}$ intersects $P_r$ then it intersects $d$. There is at most one such intersection point:


\begin{lemma}
  \label{lem:bst_intersect_d}
  The bisector $b_{st}$ intersects $d$ in at most one point $w$.
\end{lemma}

\begin{proof}
  Assume, by contradiction, that $b_{st}$ intersects $d$ in two points $w_1$
  and $w_2$, with $w_1$ above $w_2$. See
  Fig.~\ref{fig:intersection_diagonal}(b). Note that by
  Lemma~\ref{lem:bisector_prop}, $b_{st}$ cannot intersect $\partial P_\ell$,
  and thus $d$, in more than two points. Thus, the part of $b_{st}$ that lies
  in $P_\ell$ between $w_1$ and $w_2$ does not intersect $\partial
  P_\ell$. Observe that this implies that the region $R$ enclosed by this part
  of the curve, and the part of the diagonal from $w_1$ to $w_2$
  (i.e. $\overline{w_1w_2}$) is empty. Moreover, since the shortest paths from
  $t$ to $w_1$ and to $w_2$ intersect $b_{st}$ only once
  (Lemma~\ref{lem:bisector_sp_intersect}) region $R$ contains the shortest
  paths $\geod(t,w_1)=\overline{tw_1}$ and $\geod(t,w_2)=\overline{tw_2}$.

  Since $s$ has the same geodesic distance to $w_1$ and $w_2$ as $t$, $s$ must
  lie in the intersection $X$ of the disks $D_i$ with radius $\|tw_i\|$
  centered at $w_i$, for $i \in 1,2$. It now follows that $s$ lies in one of
  the connected sets, or ``pockets'', of $X \setminus R$. Assume without loss
  of generality that it lies in a pocket above $t$ (i.e.~$s_y > t_y$). See
  Fig.~\ref{fig:intersection_diagonal}. We now again use
  Lemma~\ref{lem:bisector_sp_intersect}, and get that $\geod(s,w_2)$ intersects
  $b_{st}$ only once, namely in $w_2$. It follows that the shortest path from
  $s$ to $w_2$ has to go around $R \ni t$, and thus has length strictly larger
  than $\|tw_2\|$. Contradiction.
\end{proof}

Since $b_{st}$ intersects $d$ only once (Lemma~\ref{lem:bst_intersect_d}), and
there is a point of $b_{st}$ on $\geod(s,t) \subset P_\ell$, it follows that
there is at most one point $z$ where $b_{st}$ intersects $\partial P_r$ the
outer boundary of $P_r$, i.e.~$\partial P_r \setminus d$. Observe that
therefore $z$ is a corner of the pseudo-triangle $\hat{\P}(z,s,t)$, and that
$\F(z,s,t) \subseteq \hat{\P}(z,s,t)$. Let $\br_{st}=b_{st} \cap P_r$ and
orient it from $w$ to $z$. We assign $b_{st}$ the same orientation.

\begin{lemma}
  \label{lem:bisector_in_funnel}
  (i) The bisector $b_{st}$ does not intersect $\geod(s,z)$ or $\geod(t,z)$ in
  any point other than $z$. (ii) The part of the bisector $b_{st}$ that lies in
  $P_r$ is contained in $\F(z,s,t)$.
\end{lemma}

\begin{proof}
  By Lemma~\ref{lem:bisector_sp_intersect} the shortest path from $s$ to any
  point $v \in P$, so in particular to $z$, intersects $b_{st}$ in at most one
  point. Since, by definition, $z$ lies on $b_{st}$, the shortest path
  $\geod(s,z)$ does not intersect $b_{st}$ in any other point. The same applies
  for $\geod(t,z)$, thus proving (i). For (ii) we observe that any internal
  point of $\geod(s,z)$ is closer to $s$ than to $t$, and any internal point of
  $\geod(t,z)$ closer to $t$ than to $s$. Thus, $\geod(s,z)$ and $\geod(t,z)$
  must be separated by $b_{st}$. It follows that $b_{st} \cap P_r$ lies inside
  $\F(z,s,t)$.
\end{proof}

\begin{lemma}
  \label{lem:extension_segments}
  All vertices of $\br_{st}$ lie on extension segments of the vertices
  in the pseudo-triangle $\hat{\P}(z,s,t)$.
\end{lemma}

\begin{proof}
  Assume by contradiction that $v \neq w$ is a vertex of
  $\br_{st}=b_{st} \cap P_r$ that is not defined by an extension segment of a
  vertex in $\hat{\P}(z,s,t)$. Instead, let $e \in E_s$ be the extension
  segment containing $v$, and let $u \in P \setminus \hat{\P}(z,s,t)$ be the
  starting vertex of $e$. So $\geod(s,v)$ has $u$ as its last
  internal vertex.

  \begin{figure}[tb]
    \centering
    \includegraphics[page=4]{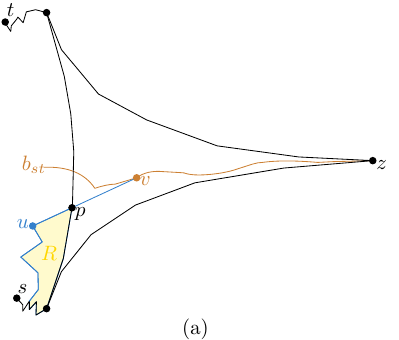}
    \quad
    \includegraphics[page=3]{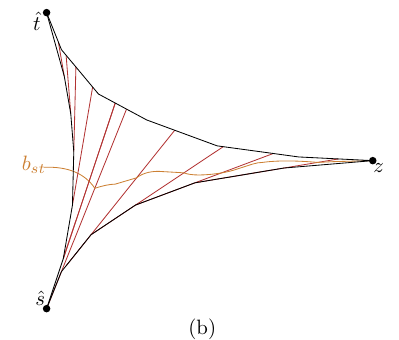}
    \caption{(a) Point $v$ lies inside $\hat{\P}(z,s,t)$, so a shortest path
      from $s$ to $v$ that uses a vertex $u$ outside of $\hat{\P}(z,s,t)$
      intersects $\partial\hat{\P}(z,s,t)$ in a point $p$. This either
      yields two distinct shortest paths from $s$ to $p$, or requires the
      shortest path from $s$ to $p$ via $u$ to intersect $b_{st}$ twice. Both
      yield a contradiction. (b) The extension segments in $F_s^t$ are all pairwise disjoint,
    and end at the chain from $t$ to $z$.}
    \label{fig:extension_segments}
  \end{figure}

  By Lemma~\ref{lem:bisector_in_funnel}, $\br_{st}$ is contained in $\F(z,s,t)$
  and thus in $\hat{\P}(z,s,t)$. Hence, $v \in \hat{\P}(z,s,t)$. Since
  $v \in \hat{\P}(z,s,t)$, and $u \not\in \hat{\P}(z,s,t)$ the shortest path from
  $s$ to $v$ intersects $\partial \hat{\P}(z,s,t)$ in some point $p$. See
  Fig.~\ref{fig:extension_segments}(a). We then distinguish two cases: either $p$
  lies on $\geod(s,z) \cup \geod(s,t)$, or $p$ lies on $\geod(t,z)$.

  In the former case this means there are two distinct shortest paths between
  $s$ and $p$, that bound a region $R$ that is non-empty, that is, it has
  positive area. Note that this region exists, even if $u$ lies on the shortest
  path from $s$ to its corresponding corner $\hat{s}$ in $\hat{\P}(z,s,t)$ but
  not on $\hat{\P}(z,s,t)$ itself (i.e.
  $u \in \geod(s,t)\cup\geod(s,z) \setminus \hat{\P}(z,s,t)$. Since $P$ is a
  simple polygon, this region $R$ is empty of obstacles, and we can shortcut
  one of the paths to $p$. This contradicts that such a path is a shortest
  path.

  In the latter case the point $p$ lies on $\geod(t,z)$, which means that it is
  at least as close to $t$ as it is to $s$. Since $s$ is clearly closer to $s$
  than to $t$, this means that the shortest path from $s$ to $v$ (that visits
  $u$ and $p$) intersects $b_{st}$ somewhere between $s$ and $p$. Since it again
  intersects $b_{st}$ at $v$, we now have a contradiction: by
  Lemma~\ref{lem:bisector_sp_intersect}, any shortest path from $s$ to $v$
  intersects $b_{st}$ at most once. The lemma follows.
\end{proof}

Let $F_s^t = e_1,..,e_g$ denote the extension segments of the vertices of
$\geod(t,s)$ and $\geod(s,z)$, ordered along $\hat{\P}(z,s,t)$, and clipped to
$\hat{\P}(z,s,t)$. See Fig.~\ref{fig:bisector_in_funnel}(b). We define $F_t^s$
analogously.

\begin{lemma}
  \label{lem:extension_segments_F}
  All vertices of $\br_{st}$ lie on clipped extension segments in
  $F_s^t \cup F_t^s$.
\end{lemma}

\begin{proof}
  By Lemma~\ref{lem:extension_segments} all vertices of $b_{st}$ in $P_r$ lie
  on $\hat{\P}(z,s,t)$. Furthermore, by Lemma~\ref{lem:bisector_in_funnel} all
  these vertices lie in $\F(z,s,t)$. Hence, it suffices to clip all extension
  segments to $\hat{\P}(z,s,t)$ (or even $\F(z,s,t)$). For all vertices on
  $\geod(t,z)$ the extension segments (with respect to $s$) are disjoint from
  $\hat{\P}(z,s,t)$. It follows that for site $s$, only the clipped extension
  segments from vertices on $\geod(s,t)$ and $\geod(s,z)$ are
  relevant. Analogously, for site $t$, only the clipped extension segments on
  $\geod(s,t)$ and $\geod(t,z)$ are relevant.
\end{proof}

\begin{observation}
  \label{obs:extension_segments_pw_disjoint}
  The extension segments in $F_s^t$ are all pairwise disjoint, start on
  $\geod(s,t)$ or $\geod(s,z)$, and end on $\geod(t,z)$.
\end{observation}


\begin{figure}[tb]
  \centering
  \includegraphics[page=3]{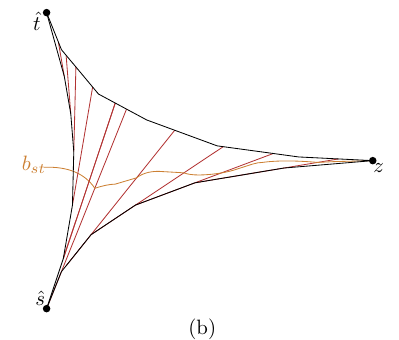}
  \caption{The extension segments in $F_s^t$ are all pairwise disjoint,
    and end at the chain from $t$ to $z$.}
  \label{fig:pseudo-triangle_slopes}
\end{figure}

By Corollary~3.29 of Aronov~\cite{aronov1989geodesic} every (clipped) extension
segment in $r \in F_s^t \cup F_t^s$ intersects $b_{st}$ (and thus $\br_{st}$)
at most once. Therefore, every such extension segment $r$ splits the bisector
in two. Together with Lemma~\ref{lem:extension_segments_F} and
Observation~\ref{obs:extension_segments_pw_disjoint} this now give us
sufficient information to efficiently binary search among the vertices of
$\br_{st}$ when we have (efficient) access to $\hat{\P}(z,s,t)$.

\begin{lemma}
  \label{lem:intersect_bst}
  Consider extension segments $e_i$ and $e_j$, with $i \leq j$, in
  $F_s^t$. If $e_i$ intersects $b_{st}$ then so does $e_j$.
\end{lemma}

\begin{proof}
  It follows from Lemma~\ref{lem:bisector_in_funnel} that $b_{st}$ intersects
  $\partial \hat{\P}(z,s,t)$ only in $z$ and in a point $w'$ on $\geod(s,t)$. Thus,
  $b_{st}$ partitions $\hat{\P}(z,s,t)$ into an $s$-side, containing $\geod(s,z)$,
  and a $t$-side, containing $\geod(t,z)$. Since the extension segments also
  partition $\hat{\P}(z,s,t)$ it then follows that the extension segments in
  $F_s^t$ intersect $b_{st}$ if and only if their starting point lies in
  the $s$-side and their ending point lies in the $t$-side. By
  Observation~\ref{obs:extension_segments_pw_disjoint} all segments in
  $F_s^t$ end on $\geod(t,z)$. Hence, they end on the $t$-side. We
  finish the proof by showing that if $e_i \in F_s^t$ starts on the
  $s$-side, so must $e_j \in F_s^t$, with $j \geq i$.

  The extension segments of vertices in $\geod(s,z)$ trivially have their start
  point on the $s$-side. It thus follows that they all intersect $b_{st}$. For
  the extension segments of vertices in $\geod(s,t)$ the ordering is such that
  the distance to $s$ is monotonically decreasing. Hence, if $e_i$ intersects
  $b_{st}$, and thus starts on the $s$-side, so does $e_j$, with $j \geq i$.
\end{proof}

\begin{lemma}
  \label{lem:fan_intersects}
  Consider extension segments $e_i$ and $e_j$, with $i \leq j$, in $F_s^t$. If
  $e_i$ intersects $\br_{st}$ then so does $e_j$.
\end{lemma}

\begin{proof}
  From Lemma~\ref{lem:intersect_bst} it follows that if $e_i$ intersects
  $b_{st}$ then so does $e_j$, with $j \geq i$. So, we only have to show that
  if $e_i$ intersects $b_{st}$ in $P_r$ then so does $e_j$. Since the
  extension segments in $F_s^t$ are pairwise disjoint, it follows that
  if $e_i$ intersects $b_{st}$, say in point $p$ then $e_j$, with $j \geq i$
  must intersect $b_{st}$ on the subcurve between $p$ and $z$. Since $b_{st}$
  intersects $d$ at most once (Lemma~\ref{lem:bst_intersect_d}), and
  $p \in P_r$, it follows that this part of the curve, and thus its
  intersection with $e_j$, also lies in $P_r$.
\end{proof}

\begin{corollary}
  \label{cor:suffix}
  The segments in $F_s^t$ that define a vertex in $\br_{st}$ form a suffix
  $G_s^t$ of $F_s^t$. That is, there is an index $a$ such that
  $G_s^t=e_a,..,e_g$ is exactly the set of extension segments in $F_s^t$ that
  define a vertex of $\br_{st}$.
\end{corollary}

When we have $\hat{\P}(z,s,t)$ and the point $w$, we can find the value $a$
from Corollary~\ref{cor:suffix} in $O(\log m)$ time as follows. We binary
search along $\geod(t,s)$ to find the first vertex $u_{a'}$ such that $u_{a'}$
is closer to $s$ then to $t$. For all vertices after $u_{a'}$, its extension
segment intersects $b_{st}$ in $\hat{\P}(z,s,t)$. To find the first segment
that intersects $b_{st}$ in $P_r$, we find the first index $a \geq a'$
for which the extension segment intersects $d$ below $w$. In total this takes
$O(\log m)$ time.
\begin{wrapfigure}[9]{r}{0.35\textwidth}
  \centering
  \vspace{-.75\baselineskip}
  \includegraphics{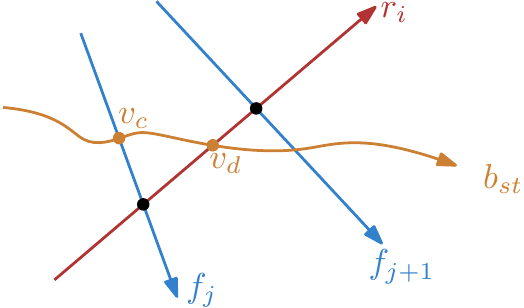}
  \caption{The bisector vertex $v_c$ on $f_j$ occurs before $v_d$ on $r_i$.}
  \label{fig:intersection_order}
\end{wrapfigure}

Let $G_s^t = r_1,..,r_{g'} = e_a,..,e_g$ be the ordered set of extension
segments that intersect $\br_{st}$. Similarly, let $G_t^s=f_1,..,f_{h'}$ be the
suffix of extension segments from $F_t^s$ that define a vertex of $\br_{st}$.

\begin{observation}
  \label{obs:intersection_order}
  Let $r_i$ be an extension segment in $G_s^t$, and let $v_d$ be the vertex of
  $b_{st}$ on $r_i$. Let $f_j$ be the last extension segment in $G_t^s$ such
  that $f_j$ intersects $r_i$ in a point closer to $s$ than to $t$. See
  Fig.~\ref{fig:intersection_order}. The vertex $v_c$ of $b_{st}$ corresponding
  to $f_j$ occurs before $v_d$, that is $c < d$.
\end{observation}

\begin{proof}
  By definition of $j$ it follows that the intersection point $v_d$ of $r_i$
  and $\br_{st}$ lies between the intersection of $r_i$ with $f_j$ and
  $f_{j+1}$. See Fig.~\ref{fig:intersection_order}. Thus, the intersection
  point
\end{proof}

\begin{lemma}
  \label{lem:random_access_bst}
  Let $j$ be the number of extension segments in $G_t^s$ that intersect
  $r_i$ in a point closer to $s$ than to $t$. Then $r_i$ contains vertex
  $v_d = v_{i+j}$ of $\br_{st}$.
\end{lemma}

\begin{proof}
  It follows from Corollary~\ref{cor:suffix} and the definition of $G_s^t$ and
  $G_t^s$ that all vertices of $\br_{st}$ lie on extension segments in
  $\{r_1,..,r_{g'}\} \cup \{f_1,..,f_{h'}\}$. Together with Corollary 3.29 of
  Aronov~\cite{aronov1989geodesic} we get that every such extension segment defines exactly
  one vertex of $\br_{st}$. Since the bisector intersects the segments
  $r_1,..,r_{g'}$ in order, there are exactly $i-1$ vertices of $\br_{st}$
  before $v_d$, defined by the extension segments in $G_s^t$. Let $f_\ell$ be
  the last extension segment in $G_t^s$ that intersects $r_i$ in a point closer
  to $s$ than to $t$. Observation~\ref{obs:intersection_order} gives us that
  this extension segment defines a vertex $v_c$ of $b_{st}$ with $c < d$. We
  then again use that $\br_{st}$ intersects the segments $f_1,..,f_{h'}$ in
  order, and thus $\ell = j$. Hence, $v_d$ is the $(i+j)^\text{th}$ vertex of
  $\br_{st}$.
\end{proof}

It follows from Lemma~\ref{lem:random_access_bst} that if we have
$\hat{\P}(z,s,t)$ and we have efficient random access to its vertices, we also
have efficient access to the vertices of the bisector $\br_{st}$. Next, we
argue with some minor augmentations the preprocessing of $P$ into a two-point
query data structure by Guibas and Hershberger gives us such access.




\subparagraph{Accessing $\hat{\P}(z,s,t)$.} The data structure of Guibas and
Hershberger can return the shortest path between two query points $p$ and $q$,
represented as a balanced
tree~\cite{guibas1989query,hershberger_new_1991}. This tree is essentially a
persistent balanced search tree on the edges of the path. Every node of the
tree can access an edge $e$ of the path in constant time, and the edges are
stored in order along the path. The tree is balanced, and supports
concatenating two paths efficiently. To support random access to the vertices
of $\hat{\P}(z,s,t)$ we need two more operations: we need to be able to access
the $i^\text{th}$ edge or vertex in a path, and we need to be able to find the
longest prefix (or suffix) of a shortest path that forms a convex chain. This
last operation will allow us to find the corners $\hat{s}$ and $\hat{t}$ of
$\hat{\P}(z,s,t)$. The data structure as represented by Guibas and Hershberger
does not support these operations directly. However, with two simple
augmentations we can support them in $O(\log m)$ time. In the following, we use
the terminology as used by Guibas and Hershberger~\cite{guibas1989query}.

The geodesic between $p$ and $q$ is returned as a balanced tree. The leaves of
this tree correspond to, what Guibas and Hershberger call, \emph{fundamental
  strings}: two convex chains joined by a tangent. The individual convex chains
are stored as balanced binary search trees. The internal nodes have two
or three children, and represent \emph{derived strings}: the concatenation of
the fundamental strings stored in its descendant leaves. See
Fig.~\ref{fig:shortest_path_ds} for an illustration.

\begin{figure}[tb]
  \centering
  \includegraphics{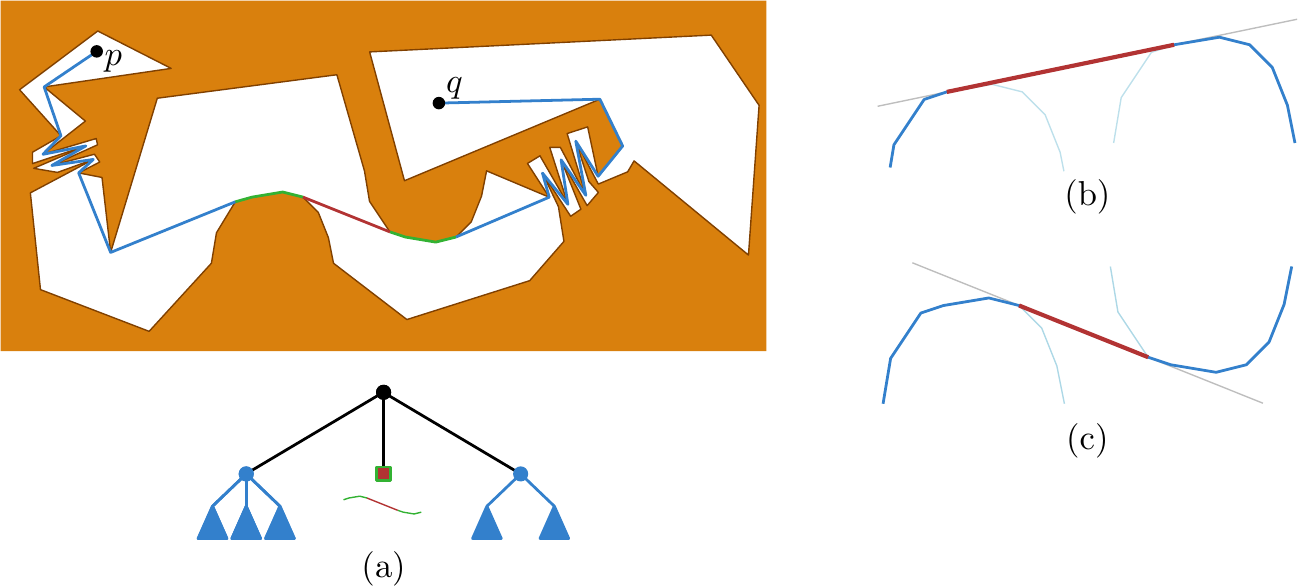}
  \caption{The data structure of Guibas and Hershberger~\cite{guibas1989query}
    can return the geodesic between two query points $p$ and $q$ as a balanced
    tree (a). The leaves of the tree correspond to fundamental strings: two
    convex chains joined by a tangent. The internal nodes represent derived
    strings: the concatenation of two or three sub-paths (strings). A
    fundamental string can be convex (b) or non-convex (c).}
  \label{fig:shortest_path_ds}
\end{figure}

To make sure that we can access the $i^\text{th}$ vertex or edge on a shortest
path in $O(\log m)$ time, we augment the trees to store subtree sizes. It is
easy to see that we can maintain these subtree sizes without affecting the
running time of the other operations.

To make sure that we can find the longest prefix (suffix) of a shortest path
that is convex we do the following. With each node $v$ in the tree we store a
boolean flag $v\mathit{.convex}$ that is true if and only if the sub path it
represents forms a convex chain. It is easy to maintain this flag without
affecting the running time of the other operations. For leaves of the tree
(fundamental strings) we can test this by checking the orientation of the
tangent with its two adjacent edges of its convex chains. These edges can be
accessed in constant time. Similarly, for internal nodes (derived strings) we
can determine if the concatenation of the shortest paths represented by its
children is convex by inspecting the $\mathit{convex}$ field of its children,
and checking the orientation of only the first and last edges of the shortest
paths. We can access these edges in constant time. This augmentation allows us
to find the last vertex $v$ of a shortest path $\geod(p,q)$ such that
$\geod(p,v)$ is a convex chain in $O(\log m)$ time. We can then obtain
$\geod(p,v)$ itself (represented by a balanced tree) in $O(\log m)$ time
by simply querying the data structure with points $p$ and $v$. Hence, we can
compute the longest prefix (or suffix) on which a shortest path forms a convex
chain in $O(\log m)$ time.


\begin{wrapfigure}[23]{r}{0.425\textwidth}
  \centering
  \includegraphics{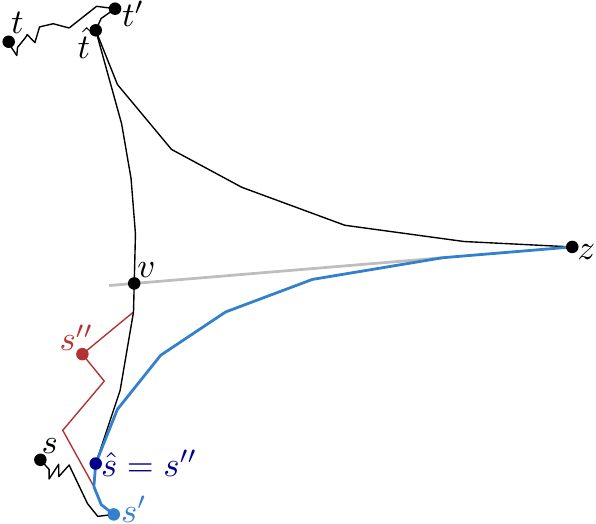}
  \vspace{-.3cm}
  \caption{We can find the first point $s'$ on $\protect\geod(s,z)$ such that
    $\protect\geod(s',z)$ is convex. When we have a point $v$ of
    $\protect\geod(s',t')$ known to be in $\hat{\P}(z,s,t)$, we can find
    $s''=\hat{s}$ (darkblue). The first point $s''$ on $\protect\geod(s',t')$
    such that $\protect\geod(s'',v)$ is convex has to lie on
    $\protect\geod(s',z)$. If this is not the case (as shown in red), then we
    can shortcut the shortest path to avoid $s''$, leading to a contradiction.
  }
  \label{fig:pseudo-triangle}
\end{wrapfigure}
Given point $z$, the above augmentation allow us to access $\hat{\P}(z,s,t)$ in
$O(\log m)$ time. We query the data structure to get the tree representing
$\geod(s,z)$, and, using our augmentations, find the longest convex suffix
$\geod(s',z)$. Similarly, we find the longest convex suffix $\geod(t',z)$ of
$\geod(t,z)$. Observe that the corners $\hat{s}$ and $\hat{t}$ of
$\hat{\P}(z,s,t)$ lie on $\geod(s',z)$ and $\geod(t',z)$, respectively
(otherwise $\geod(\hat{s},z)$ and $\geod(\hat{t},z)$ would not be convex
chains). Unfortunately, we cannot directly use the same approach to find the
part $\geod(s',t')$ that is convex, as it may both start and end with a piece
that is non-convex (with respect to $\geod(s',t')$). However, consider the
extension segment of the first edge of the shortest path from $z$ to $s$ (see
Fig.~\ref{fig:pseudo-triangle}). This extension segment intersects the shortest
path $\geod(s',t')$ exactly once in a point $v$. By construction, this point
$v$ must lie in the pseudo-triangle $\hat{\P}(z,s,t)$. Thus, we can decompose
$\geod(s',t')$ into two sub-paths, one of which starts with a convex chain and
the other ends with a convex chain. Hence, for those chains we can use the
$\mathit{convex}$ fields to find the vertices $s''$ and $t''$ such that
$\geod(s'',v)$ and $\geod(v,t'')$ are convex, and thus $\geod(s'',t'')$ is
convex. Finally, observe that $s''$ must occur on $\geod(s',z)$, otherwise we
could shortcut $\geod(s',t')$. See Fig.~\ref{fig:pseudo-triangle}. Hence, $s'' = \hat{s}$
and $t'' = \hat{t}$ are the two corners of the pseudo-triangle
$\hat{\P}(z,s,t)$. We can find $v$ in $O(\log m)$ time by a binary search on
$\geod(s',t')$. Finding the longest convex chains starting and ending in $v$
also takes $O(\log m)$ time, as does computing the shortest path
$\geod(\hat{s},\hat{t})$. It follows that given $z$, we can compute (a
representation of) $\hat{\P}(z,s,t)$ in $O(\log m)$ time.

With the above augmentations, and using Lemma~\ref{lem:random_access_bst}, we
then obtain the following result.

\begin{lemma}
  \label{lem:random_access_bst_final}
  Given the points $w$ and $z$ where $\br_{st}$ intersects $d$ and the outer
  boundary of $P_r$, respectively,
  we can access the $i^\text{th}$ vertex of $\br_{st}$ in $O(\log m)$ time.
\end{lemma}

\begin{proof}
  Recall that the data structure of Guibas and
  Hershberger~\cite{guibas1989query} reports the shortest path between query
  points $p$ and $q$ as a balanced tree. We augment these trees such that each
  node knows the size of its subtree. It is easy to do this using only constant
  extra time and space, and without affecting the other operations. We can then
  simply binary search on the subtree sizes, using
  Lemma~\ref{lem:random_access_bst} to guide the search.
\end{proof}

\subparagraph{Finding $w$ and $z$.} We first show that we can find the point $w$ where
$b_{st}$ enters $P_r$ (if it exists), and then show how to find $z$, the other
point where $b_{st}$ intersects $\partial P_r$.

\begin{lemma}
  \label{lem:finding_w}
  Finding $w$ requires $O(\log^2 m)$ time.
\end{lemma}

\begin{proof}
  Consider the geodesic distance of $s$ to diagonal $d$ as a function $f_s$,
  parameterized by a value $\lambda \in [0,1]$ along $d$. Similarly, let $f_t$
  be the distance function from $t$ to $d$. Since $b_{st}$ intersects $d$
  exactly once --namely in $w$-- the predicate
  $\hat{\P}(\lambda) = f_s(\lambda) < f_t(\lambda)$ changes from \textsc{True}
  to \textsc{False} (or vice versa) exactly once. Query the data structure of
  Guibas and Hershberger~\cite{guibas1989query} to get the funnel representing
  the shortest paths from $s$ to the points in $d$. Let $p_1,..,p_h$, with
  $h=O(m)$, be the intersection points of the extension segments of vertices in
  the funnel with $d$. Similarly, compute the funnel representing the shortest
  paths from $t$ to $d$. The extension segments in this funnel intersect $d$ in
  points $q_1,..,q_k$, with $k=O(m)$. We can now simultaneously binary search
  among $p_1,..,p_h$ and $q_1,..,q_k$ to find the smallest interval $I$ bounded
  by points in $\{p_1,..,p_h,q_1,..,q_k\}$ in which $\hat{\P}$ flips from
  \textsc{True} to \textsc{False}. Hence, $I$ contains $w$. Computing the
  distance from $s$ ($t$) to some $q_i$ ($p_i$) takes $O(\log m)$ time, and
  thus we can find $I$ in $O(\log^2 m)$ time. On interval $I$ both $f_s$ and
  $f_t$ are simple hyperbolic functions consisting of a single piece, and thus
  we can compute $w$ in constant time.
\end{proof}

Consider the vertices $v_1,..,v_h$ of $P_r$ in clockwise order, where
$d=\overline{v_hv_1}$ is the diagonal. Since the bisector $b_{st}$ intersects
the outer boundary of $P_r$ in only one point, there is a vertex $v_a$ such
that $v_1,..,v_a$ are all closer to $t$ than to $s$, and $v_{a+1},..,v_h$ are
all closer to $s$ than to $t$. We can thus find this vertex $v_a$ using a
binary search. This takes $O(\log^2 m)$ time, as we can compute
$\geodlen(s,v_i)$ and $\geodlen(t,v_i)$ in $O(\log m)$ time. It then follows
that $z$ lies on the edge $\overline{v_a,v_{a+1}}$. We can find the exact
location of $z$ using a similar approach as in Lemma~\ref{lem:finding_w}. This
takes $O(\log^2 m)$ time. Thus, we can find $z$ in $O(\log^2 m)$ time. We
summarize our results from this section in the following theorem.


\begin{theorem}
  \label{thm:represent_bisector}
  Let $P$ be a simple polygon with $m$ vertices that is split into $P_\ell$ and
  $P_r$ by a diagonal $d$. The polygon $P$ can be preprocessed in $O(m)$ time,
  so that for any pair of points $s$ and $t$ in $P_\ell$, a representation of
  $\br_{st}=b_{st} \cap P_r$ can be computed in $O(\log^2 m)$ time. This
  representation supports accessing any of its vertices in $O(\log m)$ time.
\end{theorem}


\end{document}